\newif\ifsubmit     %
\newif\ifblind      %
\newif\ifcompact    %
\newif\ifexabs      %
\newif\ifitcs       %

\newif\iflncs %

\submittrue
\iflncs
    \newcommand{\fullversion}[1]{}
    \newcommand{\subversion}[1]{#1}  
    \else
    \newcommand{\fullversion}[1]{#1}
    \newcommand{\subversion}[1]{}
\fi

\newif\ifshowabs    %
\ifitcs
  \documentclass[a4paper,UKenglish,cleveref]{lipics-v2021}
  \exabstrue
  \showabstrue
\else
  \iflncs
    \documentclass[runningheads]{llncs}\usepackage{breakcites}
  \fi
  \iflncs\else
  \documentclass[letterpaper,11pt]{article}
  \fi
  \fullversion{
  \usepackage{fullpage}
  }
  \ifexabs\else\showabstrue\fi

  \usepackage{iftex}
  \ifPDFTeX
    \usepackage[utf8]{inputenc}
    \usepackage[noTeX]{mmap}
    \usepackage[T1]{fontenc}
  \fi
  \ifLuaTeX
    \usepackage{luatex85}
    \usepackage[noTeX]{mmap}
  \fi

\fi

\usepackage{amsmath, amsfonts, amssymb, color}
\iflncs
\else \usepackage{amsthm}
\fi

\iflncs

  \spnewtheorem{claim}{Claim}{\itshape}{\rmfamily}
  
\makeatletter
\renewcommand{\@Opargbegintheorem}[4]{%
  #4\trivlist\item[\hskip\labelsep{#3#2\@thmcounterend}]}
\makeatother
\fi

\ifitcs\else
\usepackage[colorlinks]{hyperref}
  \hypersetup{linkcolor=blue,filecolor=blue,citecolor=blue,urlcolor=blue}
  \usepackage[nameinlink]{cleveref}
  \iflncs
  \crefname{claim}{claim}{claims}
  \Crefname{claim}{Claim}{Claims}
  
  \else
  \ifexabs
    \newtheorem{theorem}{Theorem}
  \else
    \newtheorem{theorem}{Theorem}[section]
  \fi
  \newtheorem{definition}[theorem]{Definition}
  \newtheorem{remark}[theorem]{Remark}
  \newtheorem{lemma}[theorem]{Lemma}
  \newtheorem{corollary}[theorem]{Corollary}
  
  \newtheorem{claim}[theorem]{Claim}
  \newtheorem*{remark*}{Remark}
  \fi
\fi

\fullversion{
\newtheorem*{theorem*}{Theorem}
\newtheorem*{lemma*}{Lemma}

}

\fullversion{
\theoremstyle{definition}}
\newtheorem{innercustomhybrid}{Hybrid}
\newenvironment{hybrid}[1]
  {\renewcommand\theinnercustomhybrid{#1}\innercustomhybrid}
  {\endinnercustomhybrid}

\usepackage{appendix}
\usepackage{algorithm}
\usepackage{algpseudocode}
\usepackage{braket}
\usepackage{comment}
\usepackage{url}
\usepackage{multicol}
\usepackage{mdframed}

\usepackage{pgfplots, tikz}
\pgfplotsset{compat=1.14}

\ifcompact
  \usepackage{enumitem}
  \setlist[description]{noitemsep, topsep=0pt}
  \setlist[enumerate]{noitemsep, topsep=0pt}
  \setlist[itemize]{noitemsep, topsep=0pt}
\fi

\ifitcs
\else
  \usepackage[style=alphabetic,minalphanames=3,maxalphanames=4,maxnames=99,backref=true]{biblatex}
  \renewcommand*{\labelalphaothers}{\textsuperscript{+}}
  \renewcommand*{\multicitedelim}{\addcomma\space}
\fi

\usepackage{soul, xcolor, xparse}
\makeatletter
  \ExplSyntaxOn
    \cs_new:Npn \white_text:n #1
    {
      \fp_set:Nn \l_tmpa_fp {#1 * .01}
      \llap{\textcolor{white}{\the\SOUL@syllable}\hspace{\fp_to_decimal:N \l_tmpa_fp em}}
      \llap{\textcolor{white}{\the\SOUL@syllable}\hspace{-\fp_to_decimal:N \l_tmpa_fp em}}
    }
    \NewDocumentCommand{\whiten}{ m }
    {
      \int_step_function:nnnN {1}{1}{#1} \white_text:n
    }
  \ExplSyntaxOff
  
  \NewDocumentCommand{ \varul }{ D<>{5} O{0.2ex} O{0.1ex} +m } {%
    \begingroup
    \setul{#2}{#3}%
    \def\SOUL@uleverysyllable{%
      \setbox0=\hbox{\the\SOUL@syllable}%
      \ifdim\dp0>\z@
      \SOUL@ulunderline{\phantom{\the\SOUL@syllable}}%
      \whiten{#1}%
      \llap{%
        \the\SOUL@syllable
        \SOUL@setkern\SOUL@charkern
      }%
      \else
      \SOUL@ulunderline{%
        \the\SOUL@syllable
        \SOUL@setkern\SOUL@charkern
      }%
      \fi}%
    \ul{#4}%
    \endgroup
  }
\makeatother

\newcommand{\E}{\mathop{\mathbb{E}}}
\newcommand{\Z}{\mathbb{Z}}
\newcommand{\I}{\mathbb{I}}
\newcommand{\Cp}{\mathbb{C}}

\newcommand{\As}{\mathcal{A}}
\newcommand{\Bs}{\mathcal{B}}
\newcommand{\Cs}{\mathcal{C}}
\newcommand{\Ds}{\mathcal{D}}

\newcommand{\Hs}{\mathcal{H}}
\newcommand{\Ss}{\mathcal{S}}

\newcommand{\A}{\mathsf{A}}
\newcommand{\B}{\mathsf{B}}
\newcommand{\C}{\mathsf{C}}

\newcommand{\adv}{\mathsf{adv}}

\newcommand{\F}{\mathbb{F}}

\newcommand{\prg}{{\sf PRG}}

\newcommand{\enc}{{\sf Enc}}
\newcommand{\dec}{{\sf Dec}}

\newcommand{\gen}{{\sf Gen}}

\newcommand{\can}{{\sf Can}}

\newcommand{\cD}{{\mathcal{D}}}

\newcommand{\cE}{{\mathcal{E}}}

\newcommand{\ati}{{\sf ATI}}

\newcommand{\cP}{{\mathcal{P}}}

\newcommand{\MOE}{unclonable }

\newcommand{\sk}{{\sf sk}}

\newcommand{\poly}{{\sf poly}}
\newcommand{\negl}{{\sf negl}}

\DeclareMathOperator{\Tr}{Tr}

\newcommand{\abs}[1]{{\left\lvert#1\right\rvert}}

\newcommand{\canonicalset}{\mathsf{CS}}
\newcommand{\brackets}[1]{\!\left( #1 \right)}
\newcommand{\bracketsC}[1]{\left\{ #1 \right\}}
\newcommand{\bracketsS}[1]{\left[ #1 \right]}

\newcommand{\fclass}{\mathcal{F}}
\newcommand{\from}{\leftarrow}
\newcommand{\density}[1]{\Ds(\Hs_{#1})}
\newcommand{\distrclass}{\mathfrak{D}_X}
\newcommand{\alice}{\As}
\newcommand{\bob}{\Bs}
\newcommand{\charlie}{\Cs}
\newcommand{\abc}{(\As, \Bs, \Cs)}
\newcommand{\aliceprime}{\As'}
\newcommand{\bobprime}{\Bs'}
\newcommand{\charlieprime}{\Cs'}
\newcommand{\abcprime}{(\As', \Bs', \Cs')}
\newcommand{\pirateexp}[4]{\mathsf{PirExp}^{#1,#2}_{#3,#4}}
\newcommand{\distr}{\Ds}
\newcommand{\prob}{\Pr}
\newcommand{\bit}{\{0,1\}}
\newcommand{\trivialprob}{p^\mathsf{triv}}

\newcommand{\given}{\mid}
\newcommand{\repryA}{G_{y}^A}
\newcommand{\oracleone}{G}
\newcommand{\oracletwo}{H}
\newcommand{\copyprotect}{\mathsf{CopyProtect}}
\newcommand{\cp}{\copyprotect}
\newcommand{\eval}{\mathsf{Eval}}
\newcommand{\secparam}{\lambda}
\newcommand{\ketbra}[2]{| #1 \rangle \langle #2 |}
\newcommand{\U}{\mathcal{U}}
\newcommand{\subspace}{{\cal S}}
\newcommand{\tracedist}[2]{\mathsf{TD}\brackets{#1, #2}}
\newcommand{\Ext}{\mathsf{Ext}}
\newcommand{\distrmarg}{\mathfrak{D}}
\newcommand{\sati}{\mathsf{SATI}}
\newcommand{\weight}{w}
\newcommand{\bad}{{\sf bad}}

\ifsubmit
    \newcommand{\luowen}[1]{}
    \newcommand{\qipeng}[1]{}
    \newcommand{\xingjian}[1]{}
    \newcommand{\fatih}[1]{}
    \newcommand{\pnote}[1]{}
\else
    \newcommand{\luowen}[1]{{\color{magenta} Luowen: #1}}
    \newcommand{\qipeng}[1]{{\color{red} Qipeng: #1}}
    \newcommand{\xingjian}[1]{{\color{cyan} Xingjian: #1}}
    \newcommand{\fatih}[1]{{\color{magenta} Fatih: #1}}
    \newcommand{\pnote}[1]{{\color{blue} P: #1}}
\fi

\mdfdefinestyle{figstyle}{ 
  linecolor=black!7, 
  backgroundcolor=black!7, 
  innertopmargin=10pt, 
  innerleftmargin=25pt, 
  innerrightmargin=25pt, 
  innerbottommargin=10pt 
}
\newenvironment{gamespec}{
  \begin{mdframed}[style=figstyle]}{
  \end{mdframed}}

\title{On the Feasibility of Unclonable Encryption, and More} 

\iflncs\else
\newcommand{\email}[1]{\href{mailto:#1}{\texttt{#1}}}
\fi
\author{
Prabhanjan Ananth\footnote{University of California, Santa Barbara. Email: \email{prabhanjan@cs.ucsb.edu}}
\and 
Fatih Kaleoglu\footnote{University of California, Santa Barbara. Email: \email{kaleoglu@ucsb.edu}}
\and 
Xingjian Li\footnote{Tsinghua University. Email: \email{lixj18@mails.tsinghua.edu.cn}}
\and 
Qipeng Liu\footnote{Simons Institute for the Theory of Computing. Email: \email{qipengliu0@gmail.com}} 
\and
Mark Zhandry\footnote{NTT Research \& Princeton University. Email: \email{mzhandry@gmail.com}}
}
\date{}

\begin{document}

\maketitle

\begin{abstract}
\noindent Unclonable encryption, first introduced by Broadbent and Lord (TQC'20), is a one-time encryption scheme with the following security guarantee: any non-local adversary $(\alice,\bob,\charlie)$ cannot simultaneously distinguish encryptions of two equal length messages. This notion is termed as unclonable indistinguishability. Prior works focused on achieving a weaker notion of unclonable encryption, where we required that any non-local adversary $(\alice,\bob,\charlie)$ cannot simultaneously recover the entire message $m$. Seemingly innocuous, understanding the feasibility of encryption schemes satisfying unclonable indistinguishability (even for 1-bit messages) has remained elusive.
\par We make progress towards establishing the feasibility of unclonable encryption. 
\begin{itemize} 
\item We show that encryption schemes satisfying unclonable indistinguishability exist unconditionally in the quantum random oracle model. 
\item Towards understanding the necessity of oracles, we present a negative result stipulating that a large class of encryption schemes cannot satisfy unclonable indistinguishability. 
\item Finally, we also establish the feasibility of another closely related primitive: copy-protection for single-bit output point functions. Prior works only established the feasibility of copy-protection for multi-bit output  point functions or they achieved constant security error for single-bit output point functions.  
\end{itemize}
\end{abstract}

\section{Introduction}
Quantum information ushers in a new era for cryptography. Cryptographic constructs that are impossible to achieve classically can be realized using quantum information. In particular, the no-cloning principle of quantum mechanics has given rise to many wonderful primitives such as quantum money~\cite{Wie83} and its variants~\cite{AC02,Zha19,RS19}, tamper detection~\cite{Got02}, quantum copy-protection~\cite{Aar09}, one-shot signatures~\cite{AGKZ20}, single-decryptor encryption \cite{GZ20,coladangelo2021hidden}, secure software leasing~\cite{ALP21}, copy-detection~\cite{aaronsonnew} and many more.

\paragraph{Unclonable Encryption.} Of particular interest is a primitive called unclonable encryption, introduced by Broadbent and Lord~\cite{broadbentUncloneableQuantumEncryption2020}. Roughly speaking, unclonable encryption is a one-time secure encryption scheme with \emph{quantum} ciphertexts having the following security guarantee: any adversary given a ciphertext, modeled as a quantum state, cannot produce two (possibly entangled) states that both encode some information about the original message. This is formalized in terms of a splitting game. 

A splitting adversary $\abc$ first has $\alice$ receive as input an encryption of $m_b$, for two messages $m_0$ and $m_1$. $\alice$ then outputs a bipartite state to $\bob$ and $\charlie$. $\bob$ and $\charlie$ additionally receive as input the classical decryption key and respectively output $b_B$ and $b_C$. They win if $b=b_B=b_C$. Clearly, $\alice$ could give $\bob$ the entire ciphertext and $\charlie$ nothing, in which case $b_B=b$ but $b_C$ would be independent of $b$, giving an overall winning probability of $1/2$. Security therefore requires that the splitting adversary wins with probability only negligibly larger than 1/2. This security property, introduced by~\cite{broadbentUncloneableQuantumEncryption2020}, is called {\em unclonable indistinguishability}. Unclonable indistinguishability clearly implies plain semantic security, as $\alice$ could use any semantic security adversary to make a guess $b_A$ for $b$, and then simply send $b_A$ to $\bob$ and $\charlie$, who set $b_B=b_C:=b_A$.

Unclonable encryption is motivated by a few interesting applications. Firstly, unclonable encryption implies private-key quantum money. It is also useful for preventing storage attacks where malicious entities steal ciphertexts in the hope that they can decrypt them when the decryption key is compromised later. Recently, the works of~\cite{CMP20,ananth2021uncloneable} showed that unclonable encryption implies copy-protection for a restricted class of functions with computational correctness guarantees.

Despite being a natural primitive, actually constructing unclonable encryption (even for 1-bit messages!) and justifying its security has remained elusive. Prior works~\cite{broadbentUncloneableQuantumEncryption2020,ananth2021uncloneable} established the feasibility of unclonable encryption satisfying a weaker property simply called {\em unclonability}: this is modeled similar to unclonable indistinguishability, except that the message $m$ encrypted is sampled uniformly at random and both $\bob$ and $\charlie$ are expected to guess the entire message $m$. This weaker property is far less useful, and both applications listed above -- preventing storage attacks and copy-protection -- crucially rely on indistinguishability security. Moreover, unclonability does not on its own even imply plain semantic security, meaning the prior works must separately posit semantic security. %

\par The following question has been left open from prior works: 

\begin{quote}
{\em Q1. Do encryption schemes satisfying unclonable indistinguishability, exist?}
\end{quote}

\paragraph{Copy-Protection for Point Functions.} Copy-protection, first introduced by Aaronson~\cite{Aar09}, is another important primitive closely related to unclonable encryption. Copy-protection is a compiler that converts a program into a quantum state that not only retains the original functionality but also satisfies the following property: a splitting adversary $(\alice,\bob,\charlie)$ first has $\alice$ receive as input a copy-protected state that can be used to compute a function $f$. $\alice$ then outputs a bipartite state to $\bob$ and $\charlie$. As part of the security guarantee, we require that both $\bob$ and $\charlie$ should not be able to simultaneously compute $f$. 
\par While copy-protection is known to be impossible for general unlearnable functions~\cite{ALP21}, we could still hope to achieve it for simple classes of functions. Of particular interest to us is the class of point functions. A single-bit output point function is of the form $f_y(\cdot)$: it takes as input $x$ and outputs $1$ if and only if $x=y$. One could also consider the notion of multi-bit output point functions, where the function outputs a large string, rather than 0 or 1.  
\par Prior works~\cite{CMP20,ananth2021uncloneable} either focus on constructing copy-protection for {\em multi}-bit output point functions or they construct copy-protection for single-bit output point functions with constant security, rather than optimal security, where the adversary can only do negligibly better than a trivial guess. %
\par Yet another important question that has been left open from prior works is the following:  

\begin{quote}
{\em Q2. Does copy-protection for single-bit output point functions, with optimal security, exist?} 
\end{quote}

\noindent As we will see later, the techniques used in resolving Q1 will shed light on resolving Q2. Hence, we focus on highlighting challenges in resolving Q1. The reader familiar with the challenges involved in constructing unclonable encryption could skip~\Cref{sec:intro:challenges} and directly go to~\Cref{sec:ourresults}.

\subsection{Achieving Unclonable Indistinguishability: Challenges} 
\label{sec:intro:challenges}

We need to achieve a {\em one-time} secure encryption scheme for {\em 1-bit} messages satisfying unclonable indistinguishability: {\em how hard can this problem be?} Indeed one might be tempted to conclude that going from the weaker unclonability property to the stronger unclonable indistinguishability notion is a small step. The former is a search problem while the latter is a decision problem, and could hope to apply known search-to-decision reductions. As we will now explain, unfortunately this intuition is false, due both to the effects of quantum information and also to the fact that unclonable encryption involves multiple interacting adversaries.
\begin{itemize}
    \item Recall that in an unclonable encryption scheme, the secret key is revealed to both $\bob$ and $\charlie$. As a consequence, the secret information of any underlying cryptographic tool we use to build unclonable encryption could be revealed. For example, consider the following construction: to encrypt $m \in \{0,1\}$, compute $(r, {\sf PRF}(k,r) \oplus m)$, where $k \xleftarrow{\$} \{0,1\}^{\secparam}$ is the pseudorandom function key and $r \xleftarrow{\$} \{0,1\}^{\secparam}$ is a random tag. In the security experiment, the secret key, namely $k$, will be revealed to both $\bob$ and $\charlie$. This restricts the type of cryptographic tools we can use to build unclonable encryption.
    
    \item Another challenge is to perform security reductions. Typically, we use the adversary to come up with a reduction that breaks a cryptographic game that is either conjectured to be or provably hard. However, this is tricky when there are two adversaries, $\bob$ and $\charlie$. Which of the two adversaries do we use to break the underlying game? Suppose we decide to use $\bob$ to break the game. For all we know, $\alice$ could have simply handed over the ciphertext it received to $\bob$ and clearly, $\bob$ cannot be used to break the underlying game. Even worse, Alice can send a superposition of ${\cal B}$ getting the ciphertext and ${\cal C}$ receiving nothing v.s. ${\cal C}$ receiving the ciphertext and ${\cal B}$ getting nothing. 

    \item Even if we somehow manage to achieve unclonable indistinguishability for 1-bit messages, it is a priori unclear how to achieve unclonable indistinguishability for multi-bit messages. In classical cryptography, the standard transformation goes from encryption of 1-bit messages to encryption of multi-bit messages via a hybrid argument. This type of argument fails in the setting of unclonable encryption. Let us illustrate why: suppose we encrypt a 2-bit message $m=m_1||m_2$ by encrypting 1-bit messages $m_1$ and $m_2$, denoted respectively by $\rho_1$ and $\rho_2$. This scheme is unfortunately insecure. An encryption of $11$ can be (simultaneously) distinguished from an encryption of $00$ by a non-local adversary $(\alice,\bob,\charlie)$: $\alice$ can send $\rho_1$ to $\bob$ and $\rho_2$ to $\charlie$. Since, both $\bob$ and $\charlie$ receive the secret key, they can check whether the underlying message was 1 or 0. 
    
    \item A recent result by Majenz, Schaffner and Tahmasbi~\cite{majenzLimitationsUncloneableEncryption2021} explores the difficulties in constructing unclonable encryption schemes. They show that any unclonable encryption scheme satisfying indistinguishability property needs to have ciphertexts, when represented as density matrices, with sufficiently large eigenvalues. As a consequence, it was shown that~\cite{broadbentUncloneableQuantumEncryption2020} did not satisfy unclonable-indistinguishability property. Any unclonable encryption scheme we come up with needs to overcome the hurdles set by~\cite{majenzLimitationsUncloneableEncryption2021}.
\end{itemize}

\noindent We take an example below that concretely highlights some of the challenges explained above.

\ \\
\noindent {\em Example: Issues with using Extractors.} For instance, we could hope to use randomness extractors. To encrypt a message $m$, we output $(\rho_x,c_r,\Ext(r,x)\oplus m)$, where $\rho_x$ is an unclonable encryption of $x$ satisfying the weaker unclonability property, $c_r$ is a classical encryption of a random seed $r$, and $\Ext$ is an extractor using seed $r$. The intuition for this construction is that unclonable security implies that at least one of the two parties, say $\bob$ cannot predict $x$, and therefore $x$ has min-entropy conditioned on $\bob$'s view. Therefore, $\Ext(r,x)$ extracts bits that are statistically random against $\bob$, and thus completely hides $m$.

There are a few problems with this proposal. First, since $\alice$ generates $\bob$'s state and has access to the entire ciphertext, the conditional distribution of $x$ given Bob's view will depend on $c_r$. This breaks the extractor application, since it requires $r$ to be independent. One could hope to perform a hybrid argument to replace $c_r$ with a random ciphertext, but this is not possible: $\bob$ eventually learns the decryption key for $c_r$ and would be able to distinguish such a hybrid. This example already begins to show how the usual intuition fails.

A deeper problem is that extractor definitions deal with a single party, whereas unclonable encryption has two recipient parties. To illustrate the issue, note that it is actually \emph{not} the case that $x$ has min-entropy against one of the parties: if $\alice$ randomly sends the ciphertext to $\bob$ or $\charlie$, each one of them can predict $x$ with probability $1/2$, so the min-entropy is only 1. In such a case the extractor guarantee is meaningless. Now, in this example one can condition on the message $\alice$ sends to $\bob,\charlie$, and once conditioned it will in fact be the case that one of the two parties has high min-entropy. But other strategies are possible which break such a conditioning argument. For example, $\alice$ could send messages that are in \emph{superposition} v.s. $\bob$ getting the ciphertext (and $\charlie$ nothing) v.s. $\charlie$ getting the ciphertext (and $\bob$ nothing). By being in superposition, we can no longer condition on which party receives the ciphertext.

\subsection{Our Results}
\label{sec:ourresults}
\noindent We overcome the aforementioned challenges and make progress on addressing both questions Q1 and Q2. We start with our results on unclonable encryption before moving onto copy-protection. 

\paragraph{Unclonable Encryption.} For the first time, we establish the feasibility of unclonable encryption. Our result is in the quantum random oracle model. Specifically, we prove the following. 

\begin{theorem}[Informal]
\label{thm:intro:positiiveresult}
There exists an unconditionally secure one-time encryption scheme satisfying unclonable indistinguishability in the quantum random oracle model. 
\end{theorem}

\noindent Our construction is simple: we make novel use of coset states considered in recent works~\cite{coladangelo2021hidden}. However, our analysis is quite involved: among many other things, we make use of threshold projective implementation introduced by Zhandry~\cite{Zha19}.
\par A recent work~\cite{ananth2021uncloneable} showed a generic transformation from one-time unclonable encryption to public-key unclonable encryption\footnote{While their result demonstrates that the generic transformation preserves the unclonability property, we note that the same transformation preserves unclonable indistinguishability.}. By combining the above theorem with the generic transformation of~\cite{ananth2021uncloneable}, we obtain a public-key unclonable encryption satisfying the unclonable indistinguishability property.

\begin{theorem}[Informal]
Assuming the existence of post-quantum public-key encryption, there exists a post-quantum public-key encryption scheme satisfying the unclonable indistinguishability property in the quantum random oracle model.
\end{theorem}

\par It is natural to understand whether we can achieve unclonable encryption in the plain model. Towards understanding this question, we show that a class of unclonable encryption schemes, that we call {\em deterministic} schemes, are impossible to achieve. By `deterministic', we mean that the encryptor is a unitary $U$ and the decryptor is $U^{\dagger}$. Moreover, the impossibility holds even if the encryptor and the decryptor are allowed to run in exponential time! 
\par In more detail, we show the following. 

\begin{theorem}[Informal]
There do not exist unconditionally secure deterministic one-time encryption schemes satisfying the unclonable indistinguishability property.
\end{theorem}

\noindent In light of the fact that any classical one-time encryption scheme can be made deterministic without loss of generality\footnote{We can always include the randomness used in the encryption as part of the secret key.}, we find the above result to be surprising.
An interesting consequence of the above result is an alternate proof that the conjugate encryption scheme of~\cite{broadbentUncloneableQuantumEncryption2020} does not satisfy unclonable indistiguishability\footnote{It is easy to see why conjugate encryption of multi-bit messages is insecure. The insecurity of conjugate encryption of 1-bit messages was first established by~\cite{majenzLimitationsUncloneableEncryption2021} .}. This was originally proven by~\cite{majenzLimitationsUncloneableEncryption2021}.
\par We can overcome the impossibility result by either devising an encryption algorithm that traces out part of the output register (in other words, performs non-unitary operations) or the encryption scheme is based on computational assumptions.

\paragraph{Copy-Protection for Point Functions.} We also make progress on Q2. We show that there exists copy-protection for single-bit output functions with optimal security. Prior work by Coladangelo, Majenz and Poremba~\cite{CMP20} achieved a copy-protection scheme for single-bit output point functions that only achieved constant security. 
\par We show the following. 

\begin{theorem}[Informal]
There exists a copy-protection scheme for single-bit output point functions in the quantum random oracle model. 
\end{theorem}

\noindent While there are generic transformations from unclonable encryption to copy-protection for point functions explored in the prior works~\cite{CMP20,ananth2021uncloneable}, the transformations only work for multi-bit point functions. Our construction extensively makes use of the techniques for achieving unclonable encryption (\Cref{thm:intro:positiiveresult}). Our result takes a step closer in understanding the classes of functions for which the feasibility of copy-protection can be established. 

\subsection{Organization} 
\label{sec:organization}
The rest of the paper is organized as follows. In \Cref{sec:prelim}, we cover all the necessary preliminaries, including Jordan's lemma, measuring success probability of a quantum adversary and the definitions of unclonable encryption schemes. Followed by \Cref{sec:moe_coset}, we recall coset states and their properties. We introduce a new game called ``strengthened MOE games in the QROM'' and prove security in this game. This part contains the main technical contribution of our paper. In \Cref{sec:ue_scheme}, we build our unclonable encryption on the new property \pnote{what does "on the new property" here mean?}. In the final section (\Cref{sec:CP_qrom}), we present our construction for copy-protection of single-output point functions\subversion{; due to the technique similarities with \Cref{sec:moe_coset}, most of the technical proofs are moved to \Cref{sec:proofs}}. Finally, we talk about our impossibility result in \Cref{sec:imp}.

\subsection{Technical Overview}

\paragraph{Attempts based on Wiesner States.}
We start by recalling the unclonable encryption scheme proposed by Broadbent and Lord \cite{broadbentUncloneableQuantumEncryption2020}. 
The core idea is to encrypt a message $m$ under a randomly chosen secret key $x$ and encode $x$ into an unclonable quantum state $\rho_x$. 
Intuitively, for any splitting adversary $(\As, \Bs, \Cs)$, there is no way for $\As$ to split $\rho_x$ into two quantum states, such that no-communicating $\Bs$ and $\Cs$ can both 
recover enough information about $x$ to decrypt $\enc(x, m)$. 

A well-known choice of no-cloning states is the Wiesner conjugate coding (or Wiesner states for short) \cite{Wie83}. For a string $x = x_1 x_2 \cdots x_\lambda \in \{0,1\}^\lambda$, $\lambda$ bases are chosen uniformly at random, one for each $x_i$. Let $\theta_i$ denote the basis for $x_i$. If $\theta_i$ is $0$, $x_i$ is encoded under the computational basis $\{\ket 0, \ket 1\}$; otherwise, $x_i$ is encoded under the Hadamard basis $\{\ket +, \ket -\}$. The conjugate coding of $x$ under basis $\theta$ is then denoted by $\ket{x^\theta}$. 
By knowing $\theta$, one can easily recover $x$ from the Wiesner state. 

The unclonability of Wiesner states is well understood and characterized by \emph{monogamy-of-entanglement games} (MOE games) in \cite{Tomamichel_2013,broadbentUncloneableQuantumEncryption2020}. In the same paper, 
Broadbent and Lord show that no strategy wins the following MOE game\footnote{This is a variant of {MOE games} discussed in \cite{Tomamichel_2013}. We will use this notation throughout the paper.} with probability more than $0.85^\lambda$.

\begin{figure}[hpt]
    \centering
    \begin{gamespec}
    \begin{itemize}
    \item A challenger samples uniformly at random $x, \theta \in \{0,1\}^\lambda$ and sends $\ket{x^\theta}$ to $\As$. 
    \item $\As$ taking the input from the challenger, produces a bipartite state to $\Bs$ and $\Cs$. 
    \item The non-communicating $\Bs$ and $\Cs$ then additionally receive the secret basis information $\theta$ and make a guess $x_{\Bs}, x_{\Cs}$ for $x$ respectively. 
    \item The splitting adversary $(\As, \Bs, \Cs)$ wins the game if and only if $x_{\Bs} = x_{\Cs} = x$. 
    \end{itemize}
    \end{gamespec}
    \caption{MOE Games for Wiesner States.}
    \label{fig:BB84MOE}
\end{figure}

A natural attempt to construct unclonable encryption schemes is by composing a one-time pad with Wiesner states. A secret key is the basis information $\theta \in \{0,1\}^n$. 
An encryption algorithm takes the secret key $\theta$ and a plaintext $m$, it samples an $x \in \{0,1\}^n$ and outputs $m \oplus x$ together with the Wiesner conjugate coding of $x$, i.e. $\ket {x^\theta}$. On a high level, no split adversaries can both completely recover $x$, thus it is impossible for them to both recover the message $m$. 
However, such a scheme can never satisfy the stronger security: unclonable indistinguishability. 
Recall that unclonable indistinguishability requires either $\Bs$ or $\Cs$ can not distinguish whether the ciphertext is an encryption of message $m_0$ or $m_1$. 
Broadbent and Lord observe that although it is hard for $\Bs$ and $\Cs$ to recover the message completely, they can still recover half of the message and hence simultaneously distinguish with probability $1$. 

Towards unclonable indistinguishability, they introduce a random oracle $H:\{0, 1\}^\lambda \times \{0,1\}^\lambda \to \{0,1\}^n$ in their construction (\Cref{fig:UE_BL}).
If an adversary can distinguish between $m_0 \oplus H(\alpha, x)$  and $m_1 \oplus H(\alpha, x)$, it must query $H(\alpha, x)$ at some point; hence, one can extract $x$ from this adversary by measuring a random query. Following the same reasoning, one may hope to base the security (of \Cref{fig:UE_BL}) on the MOE games (of \Cref{fig:BB84MOE}), by extracting $x$ from both parties.

\begin{figure}[!hbt]
    \centering
    \begin{gamespec}
    \begin{description}
    \item $\gen(1^\lambda)$: on input $\lambda$, outputs uniformly random $(\alpha, \theta) \in \{0,1\}^{2 \lambda}$. 
    \item $\enc^H((\alpha, \theta), m)$: samples $x \in \{0,1\}^\lambda$, outputs $(\ket{x^\theta}, m \oplus H(\alpha, x))$.
    \item $\dec^H((\alpha, \theta), (\ket{x^\theta}, c))$: recovers $x$ from $\ket{x^\theta}$, outputs $c \oplus H(\alpha, x)$. 
    \end{description}
    \end{gamespec}
    \caption{Unclonable Encryption by Broadbent and Lord.}
    \label{fig:UE_BL}
\end{figure}

The above idea, thought intuitive, is hard to instantiate. It will require simultaneous extraction of the secret $x$ from both $\Bs$ and $\Cs$. Since $\Bs$ and $\Cs$ can be highly entangled, a successful extraction of $x$ on $\Bs$'s register may always result in an extraction failure on the other register. Broadbent and Lord use a ``simultaneous'' variant of the so-called ``O2H'' (one-way-to-hiding) lemma \cite{unruh15revocable} to prove their scheme satisfies unclonable indistinguishability for un-entangled adversaries $\Bs, \Cs$, or messages with constant length. The unclonable indistinguishability for general adversaries and message spaces remains quite unknown. 

Even worse, Majenz, Schaffner, and Tahmasbi \cite{majenzLimitationsUncloneableEncryption2021} show that there is an inherent limitation to this simultaneous variant of O2H lemma. They give an explicit example that shatters the hope of proving unclonable indistinguishability of the construction in \cite{broadbentUncloneableQuantumEncryption2020} using this lemma.

\paragraph{Instantiating \cite{broadbentUncloneableQuantumEncryption2020} using Coset States.} 
Facing the above barrier, we may resort to other states possessing some forms of unclonability. 
One candidate is the so-called ``coset states'', first proposed by Vidick and Zhang \cite{vidick2021classical} in the context of proofs of quantum knowledge and later studied by Coladangelo, Liu, Liu, and Zhandry \cite{coladangelo2021hidden} for copy-protection schemes. 

A coset state is described by three parameters: a subspace $A \subseteq \mathbb{F}_2^\lambda$ of dimension $\lambda/2$ and two vectors $s, s' \in \mathbb{F}^\lambda_2$ denoting two cosets $A + s$ and $A^\perp + s'$\footnote{There are many vectors in $A+s$. In the rest of the discussion, we assume $s$ is the lexicographically smallest vector in $A+s$. Similarly for $s'$. }($A^\perp$ denotes the dual subspace of $A$); we write the state as $\ket{A_{s, s'}}$. Coset states have many nice properties, among those we only need the following:  
\begin{enumerate}
    \item Given $\ket{A_{s, s'}}$ and a classical description of subspace $A$, an efficient quantum algorithm can compute both $s$ and $s'$. 
    \item No adversary can win the MOE game (\Cref{fig:CSMOE}) for coset states with probability more than $\sqrt{e} \cdot (\cos (\pi/8))^\lambda$ (first proved in \cite{coladangelo2021hidden}).
\end{enumerate}
\begin{figure}[!hbt]
    \centering
    \begin{gamespec}
    \begin{itemize}
    \item A challenger samples uniformly at random a subspace $A \subseteq \mathbb{F}_2^\lambda$ of dimension $\lambda/2$ $s, s' \in \mathbb{F}^\lambda_2$ and sends $\ket{A_{s, s'}}$ to $\As$. 
    \item $\As$ taking the input from the challenger, produces a bipartite state to $\Bs$ and $\Cs$. 
    \item The non-communicating $\Bs$ and $\Cs$ then additionally receive a classical description of the subspace $A$ and make a guess $s_{\Bs},s'_{\Bs}, s_{\Cs},s'_{\Cs}$ for $s,s'$ respectively. 
    \item The splitting adversary $(\As, \Bs, \Cs)$ wins the game if and only if $s_{\Bs} = s_{\Cs} = s, s'_{\Bs} = s'_{\Cs} = s'$. 
    \end{itemize}
    \end{gamespec}
    \caption{MOE Games for Coset States.}
    \label{fig:CSMOE}
\end{figure}

Readers may already notice the similarity between Wiesner states and coset states. If we substitute the basis information $\theta$ with $A$ and the secret $x$ with $s||s'$, we get coset states and their corresponding MOE games. Hence, we can translate the construction in \cite{broadbentUncloneableQuantumEncryption2020} using the languages of coset states. A question naturally arises: if these two kinds of states are very similar, why does replacing Wiesner states with coset states even matter? 

Indeed, they differ in one crucial place. Let us come back to Wiesner states. As shown by \cite{lutomirski2010online} in the setting of private key quantum money, given $\ket{x^\theta}$ together with an oracle $P_{x_c}, P_{x_h}$\footnote{\cite{lutomirski2010online} showed that an algorithm breaks the money scheme, given oracle access to $P^\theta_x$; $P^\theta_x$ outputs $1$ if and only if input $y = x$ under basis specified by $\theta$. One can change the algorithm so that it only needs $P_{x_c}$(and $P_{x_h}$) to break the money scheme where $P_{x_c}$ (or $P_{x_h}$) matches $y$ with $x$ on all coordinates such that $\theta_i = 0$ ($\theta_i = 1$) and outputs $1$ if they match.} that outputs $1$ only if input $y = x$, there exists an efficient quantum adversary that learns $x$ without knowing $\theta$. This further applies to the MOE games for Wiesner states: if $\As$ additionally gets oracle access to $P_{x_c}, P_{x_h}$, the MOE game is no longer secure. 

MOE games for coset states remain secure if oracles for checking $s$ and $s'$ are given. More formally, let $P_{A+s}$ be an oracle that outputs $1$ only if the input $y \in A + s$, similarly for $P_{A^\perp + s'}$. No adversary $(\As, \Bs, \Cs)$ can win the MOE games for coset states with more than some exponentially small probability in $\lambda$, even if $\As, \Bs, \Cs$ all query $P_{A+s}$ and $P_{A^\perp + s'}$ polynomially many times. We call this game \emph{MOE game for coset states with membership checking oracles}. 

We now give our construction of unclonable encryption that satisfies unclonable indistinguishability in \Cref{fig:UE}. In our construction, we also get rid of the extra input $\alpha$ in \cite{broadbentUncloneableQuantumEncryption2020} construction. We believe $\alpha$ can be similarly removed in their construction as well. Also, note that in our construction, we only require coset states and random oracles. The membership checking oracles will only be given to the adversary when we prove its security. Thus, we prove a stronger security guarantee (with membership checking oracle are given). Due to this, we can not prove the security of their construction using Wiesner states following the same idea; nonetheless, we do not know how to disprove it. We leave it as an interesting open question.

\begin{figure}[!hbt]
    \centering
    \begin{gamespec}
    \begin{description}
    \item $\gen(1^\lambda)$: on input $\lambda$, outputs uniformly random subspace $A \subseteq \mathbb{F}_2^\lambda$ of dimension $\lambda/2$. 
    \item $\enc^H(A, m)$: samples $s, s' \in \mathbb{F}_2^\lambda$\footnote{We again require $s, s'$ to be the lexicographically smallest vector in $A+s$ and $A^\perp + s'$.}, outputs $(\ket{A_{s, s'}}, m \oplus H(s, s'))$.
    \item $\dec^H(A, (\ket{A_{s, s'}}, c))$: recovers $s, s'$ from the coset state, outputs $c \oplus H(s, s')$. 
    \end{description}
    \end{gamespec}
    \caption{Our Unclonable Encryption Scheme.}
    \label{fig:UE}
\end{figure}

\paragraph{Basing Security on Reprogramming Games.}
Now we look at what property we require for coset states to establish unclonable indistinguishability. 
We will focus on the case $n=1$ (length-$1$ messages) in this section. By a sequence of standard variable substitutions, unclonable indistinguishability of our scheme can be based on the following security game in the identical challenge mode (please refer to \Cref{fig:CS_reprogram}), where each of $\Bs, \Cs$ tries to identify whether the oracle has been reprogrammed or not. We want to show any adversary $(\As, \Bs, \Cs)$ only achieves successful probability $1/2+\negl$. This ideal security matches the trivial attack: $\Bs$ gets the coset state and $\Cs$ makes a random guess, they win with probability $1/2$.

\begin{figure}[!hbt]
    \centering
    \begin{gamespec}
    \begin{itemize}
    \item $H$ be a random oracle with binary range, $H : \mathbb{F}_2^\lambda \times \mathbb{F}_2^\lambda \to \{0,1\}$. 
    
    Additionally, $\As, \Bs, \Cs$ get oracle access to $P_{A+s}$ and $P_{A^\perp+s'}$.
    
    \item A challenger samples a coset state $\ket{A_{s, s'}}$ and sends $(\ket{A_{s, s'}}, H(s, s'))$ to $\As$. 
    \item $\As$ taking the input from the challenger, has oracle access to $H_{(s,s') \to \bot}$ and produces a bipartite state to $\Bs$ and $\Cs$.  Here $H_{(s,s') \to \bot}$ is the same as $H$ except $H(s, s')$ is replaced with $\bot$\footnote{In the actual proof, $H(s, s')$ is replaced with a uniformly random $u$. Both approaches work.}. 
    \item The non-communicating $\Bs$ and $\Cs$ then receive a classical description of the subspace $A$:
    \begin{itemize}
        \item Let $H_0 := H$ be the original random oracle. 
        \item Let $H_1$ be identical to $H$, except the outcome on $(s, s')$ is flipped. 
        \item (\textbf{Identical Challenge Mode}): Flip a coin $b$, both $\Bs$ and $\Cs$ get oracle access to $H_b$.
        \item (\textbf{Independent Challenge Mode}): Flip two coins $b_\Bs, b_\Cs$, $\Bs$ has oracle access to $H_{b_\Bs}$ and  $\Cs$ gets oracle access to $H_{b_\Cs}$. 
    \end{itemize}
    \item $\Bs, \Cs$ makes a guess $b', b''$ respectively. 
    \item The adversary $(\As, \Bs, \Cs)$ wins the game if and only if $b' = b'' = b$ (in the identical challenge mode), or $b' = b_\Bs$ and $b'' = b_\Cs$ (in the independent challenge mode).
    \end{itemize}
    \end{gamespec}
    \caption{Reprogramming Games for Coset States in the QROM}
    \label{fig:CS_reprogram}
\end{figure}

Note that in the above reprogramming game (\Cref{fig:CS_reprogram}), $\As$ has no access to $H(s, s')$. This is different from unclonable indistinguishability games or MOE games. Nevertheless, we show that $\As$ never queries $(s, s')$ and thus $H(s, s')$ does not help $\As$ and thus can be safely removed by introducing a small loss. 

The security of the reprogramming games in the identical challenge mode can be reduced to the security in the independent challenge mode. 
A careful analysis of Jordan's lemma (\Cref{sec:Jordan}) is required to show such a reduction. We believe that this reduction is  non-trivial and we leave it to the last section in the overview.

The remaining is to show the security of the game in the independent challenge mode. Inspired by the work of \cite{z20} which initiates the study of measuring success probability of a quantum program, we show there is an efficient procedure that operates locally on both the entangled adversaries $(\Bs, \Cs)$ and outputs $(\Bs', p_\Bs)$, $(\Cs', p_\Cs)$ such that: (informally)
\begin{itemize}
    \item $\Bs'$ and $\Cs'$ are un-entangled\footnote{$\Bs'$ and $\Cs'$ satisfy a weaker guarantee than being un-entangled. Informally, conditioned on any event of non-negligible chance on one's side, the other party still has success probability $p_\Cs$ (or $p_\Bs$, respectively). The same analysis applies to this weaker guarantee. For ease of presentation, we assume that they are un-entangled.}. 
    \item The success probability of $\Bs'$ on guessing whether it has access to $H_0$ or $H_1$ is $p_\Bs$.
    \item The success probability of $\Cs'$ on  guessing whether it has access to $H_0$ or $H_1$ is $p_\Cs$.
    \item The expectation of $p_\Bs \cdot p_{\Cs}$ is equal to $(\Bs, \Cs)$'s success probability in the reprogramming game in the independent challenge mode. 
\end{itemize}
The above estimation procedure requires to run $\Bs'$ and $\Cs'$ on $H_0$ and $H_1$. 
In other words, the procedure should be able to reprogram $H_{(s, s') \to \bot}$ on the input $(s, s')$. Since the procedure will be used in the reduction for breaking MOE games for coset states, it should not know $s$ or $s'$, but only knows $A$ and $P_{A+s}, P_{A^\perp + s'}$. Nonetheless, we show with the membership checking oracle, such reprogramming is possible. For example, $H_1$ can be reprogrammed as follows:
\begin{align*}
    H_{1} = \begin{cases}
        \neg H(s, s')  &  Q_{s}(z) = 1 \text{ and } Q_{s'}(z') = 1 \\
        H_{(s, s') \to \bot}(z, z') & \text{Otherwise}
    \end{cases},
\end{align*}
where $Q_s$ is the point function that only outputs $1$ on $s$, similarly for $Q_{s'}$. The remaining is to show $Q_s$ (or $Q_{s'}$) can be instantiated by the classical description of $A$ and $P_{A+s}$ (or $P_{A^\perp + s'}$ respectively). $Q_s$ can be implemented by (1) check if the input $z$ is in $A+s$, (2) check if the input $z$ is the lexicographically smallest in $A+s$. Step (1) can be done via $P_{A+s}$. Step (2) can be done by knowing $A$ and some $z \in A+s$ (which is known from step (1)): one can check if there exists some lexicographically smaller $z^*$ such that $(z - z^*) \in {\sf span}(A)$; this can be done efficiently, by enumerating each coordinate and doing Gaussian elimination. 
Thus, both $Q_s$ and $Q_{s'}$ can be implemented. 

Without membership checking oracle, we do not know how to reprogram the oracle, or run the above procedure. Thus the proof fails for Wiesner states.

\medskip

Finally, we prove the security of reprogramming game in the independent challenge mode. 
If $(\As, \Bs, \Cs)$ has non-trivial success probability $1/2+\gamma$ for some large $\gamma$, the above procedure must output large $p_\Bs, p_\Cs > 1/2 + \gamma / 2$ with non-negligible probability. 
If $\Bs'$ never queries $H_0$ or $H_1$ on $(s, s')$, the best probability it can achieve is $1/2$. Thus, by measuring a random query of $\Bs'$, we can extract $s, s'$ with non-negligible probability. Similarly for $\Cs'$. This violates the MOE games for coset states with membership checking oracles, a contradiction. Therefore, the security of the reprogramming game in the independent mode is established.

\paragraph{Relating Identical Challenge Mode to Independent Challenge Mode.}

In the end, in this section, we discuss how to relate the reprogramming game in the identical challenge mode to that in the independent challenge mode. We refer the readers to the proof of \Cref{thm:moe_rom} for further details. 

We first elaborate on the above discussion for \emph{independent challenge mode}. It helps us establish the language for the presentation of \emph{identical challenge mode} and give a nice characterization of the state produced by Alice.

For a random choice of $A, s, s'$ and oracles $H_{(s, s') \to \bot}$, let $\ket{\sigma}_{\mathbf{BC}}$ be the joint quantum state shared by Bob and Charlie after Alice's stage. We additionally define projections $\Pi^B_b$ and $\Pi^C_b$ for $b \in \{0,1\}$: 
\begin{itemize}
    \item $\Pi^B_0$: Run Bob on its own register $\sigma[\mathbf{B}]$ with oracle access to $H_0$, project onto Bob outputting $0$ and rewind; 
    \item $\Pi^B_1$: Run Bob on $\sigma[\mathbf{B}]$ with oracle access to $H_1$, project onto Bob outputting $1$ and rewind.
\end{itemize}
We can similarly define $\Pi^C_0$ and $\Pi^C_1$. Namely, $\Pi^B_b$ is the projection for Bob's success on $H_b$ and $\Pi^C_b$ is the projection for Charlie's success on $H_b$.

By definition, the success probability in the independent challenge mode is:
\begin{align}\label{eq:prob_ind_game}
    {\sf Tr}\left[ \left(\frac{\Pi^B_0 + \Pi^B_1}{2}\right) \otimes \left( \frac{\Pi^C_0 + \Pi^C_1}{2} \right) \ket{\sigma}\bra{\sigma} \right].
\end{align}
Since ${(\Pi^B_0 + \Pi^B_1)}/{2}$ is a POVM, let $\{\ket{\phi_p}\}_{p \in \mathbb{R}}$ be the set of eigenvectors with eigenvalues $p \in [0,1]$\footnote{There can be multiple eigenvectors with the same eigenvalues. In the overview, we assume that eigenvalues are unique.}.  Similarly, let $\{\ket{\psi_q}\}_{q \in \mathbb{R}}$ be the set of eigenvectors with eigenvalues $q \in [0,1]$ for $(\Pi^C_0+\Pi^C_1)/2$. Therefore, we can write $\ket{\sigma}$ under the bases $\{\ket{\phi_p}\}$ and $\{\ket{\psi_q}\}$:
\begin{align*}
    \ket{\sigma} = \sum_{p, q} \alpha_{p, q} \ket{\phi_p} \ket{\psi_q}. 
\end{align*}

The analysis in the last paragraph (for independent challenge mode) can show in this setting that, $p$ and $q$ cannot be simultaneously far away from the trivial guessing probability $1/2$, i.e., for any inverse polynomial $\varepsilon$,
\begin{align*}
    \sum_{\substack{p: |p-1/2| > \varepsilon \\ q: |q-1/2| > \varepsilon}} |\alpha_{p, q}|^2 \approx 0.
\end{align*}
In other words, $\ket{\sigma}$ is very close to the summation of the following subnormalized states: 
\begin{align*}
    \ket{\sigma} = \sum_{\substack{p: |p-1/2| \leq \varepsilon}} \alpha_{p, q} \ket{\phi_p} \ket{\psi_q} + \sum_{\substack{p: |p-1/2| > \varepsilon \\ q: |q-1/2| \leq \varepsilon}} \alpha_{p, q} \ket{\phi_p} \ket{\psi_q}. 
\end{align*}
Here we simply call the first subnormalized state as $\ket{\sigma^\bad_{\Bs}}$, denoting Bob can not behave in a significantly different way from random guessing; and call second subnormalized state as $\ket{\sigma^\bad_{\Cs}}$ for Charlie. We have $\ket{\sigma} = \ket{\sigma^\bad_{\Bs}} + \ket{\sigma^\bad_{\Cs}}$. Thus, \ref{eq:prob_ind_game} is bounded by at most $1/2 + \varepsilon$ for any inverse polynomial $\varepsilon$, concluding the security in the independent challenge mode. 

The above analysis gives a characterization of $\ket{\sigma}$. Note that although the analysis is done assuming Alice, Bob and Charlie play the game in the independent challenge mode, it holds for the game in identical challenge mode as well.

\medskip

Finally, we focus on the identical challenge mode. 
The success probability in the identical challenge mode is:
\begin{align} \label{eq:prob_same_game}
    {\sf Tr}\left[ \left(\frac{\Pi^B_0 \otimes \Pi^C_0 + \Pi^B_1 \otimes \Pi^C_1}{2}\right) \ket{\sigma}\bra{\sigma} \right].%
\end{align}
By plugging $\ket{\sigma} = \ket{\sigma^\bad_{\Bs}} + \ket{\sigma^\bad_{\Cs}}$ in the above formula, \ref{eq:prob_same_game} is at most: 
\begin{align}\label{eq:prob_same_simplified}
    \frac{1}{2} + \varepsilon + \frac{1}{2}\brackets{\left| \langle \sigma^\bad_{\Bs} | \Pi^B_0 \otimes \Pi^C_0 | \sigma^\bad_{\Cs} \rangle \right| + \left| \langle \sigma^\bad_{\Bs} | \Pi^B_1 \otimes \Pi^C_1 | \sigma^\bad_{\Cs} \rangle \right| }.
\end{align}
The only difference between \ref{eq:prob_same_game} and \ref{eq:prob_ind_game} is the cross terms $\left| \langle \sigma^\bad_{\Bs} | \Pi^B_b \otimes \Pi^C_b | \sigma^\bad_{\Cs} \rangle \right|$, for $b \in \{0,1\}$. Perhaps surprisingly, we prove that the cross terms are \textbf{zero}. To show it, we prove a corollary of Jordan's lemma (see \Cref{cor:orthogonal_eigen}) that for any \emph{two} projections $\Pi_0, \Pi_1$, let $\ket{\phi_p}$ be the set of eigenvectors for $(\Pi_0 + \Pi_1)/2$; if $p + q \ne 1$ and $p \ne q$, then their cross terms $\langle \phi_p | \Pi_0 | \phi_q \rangle = \langle \phi_p | \Pi_1 | \phi_q \rangle = 0$. 
Applying this corollary to \ref{eq:prob_same_simplified}, we can show that $\left| \langle \sigma^\bad_{\Bs} | \Pi^B_b \otimes \Pi^C_b | \sigma^\bad_{\Cs} \rangle \right| = 0$ for both $b \in \{0,1\}$.
Therefore, we conclude the security in the identical challenge mode.

\subsection{Related Work}

\paragraph{Unclonable Encryption.} Broadbent and Lord~\cite{broadbentUncloneableQuantumEncryption2020} demonstrated the feasibility of unclonable encryption satisfying the weaker unclonability property. They present two constructions. The first construction based on Wiesner states achieve $0.85^{n}$-security (i.e., the probability that both ${\cal B}$ and ${\cal C}$ simultaneously guess the message is at most $0.85^{n}$), where $n$ is the length of the message being encrypted. Their second construction, in the quantum random oracle model, achieves $\frac{9}{2^n} + \negl(\secparam)$-security. In the same work, they show that any construction satisfying $2^{-n}$-unclonability implies unclonable indistinguishability property. Following Broadbent and Lord, Ananth and Kaleoglu~\cite{ananth2021uncloneable} construct public-key and private-key unclonable encryption schemes from computational assumptions. Even~\cite{ananth2021uncloneable} only achieve unclonable encryption with the weaker unclonability guarantees.
\par Majenz, Schaffner and Tahmasbi~\cite{majenzLimitationsUncloneableEncryption2021} explore the difficulties in constructing unclonable encryption schemes. In particular, they show that any scheme achieving unclonable indistinguishability should have ciphertexts with large eigenvalues. Towards demonstrating a better bound for unclonability, they also showed inherent limitations in the proof technique of Broadbent and Lord. 

\paragraph{Copy-Protection.} Copy-protection was first introduced by Aaronson~\cite{Aar09}. Recently, Aaronson, Liu, Liu, Zhandry and Zhang~\cite{aaronsonnew} demonstrated the existence of copy-protection in the presence of classical oracles. Coladangelo, Majenz and Poremba~\cite{CMP20} showed that copy-protection for multi-bit output point functions exists in the quantum random oracle model. They also showed that copy-protection for single-bit output point functions exists in the quantum random oracle model with constant security. 
\par Ananth and La Placa~\cite{ALP21} showed a conditional result that copy-protection for arbitrary unlearnable functions, without the use of any oracles, does not exist. Recently, Coladangelo, Liu, Liu and Zhandry~\cite{coladangelo2021hidden}, assuming post-quantum indistinguishability obfuscation and one-way functions,  demonstrated the first feasibility of copy-protection for a non-trivial class of functions (namely, pseudorandom functions) in the plain model. Another recent work by Broadbent, Jeffrey, Lord, Podder and Sundaram~\cite{BJLPS21} studies copy-protection for a novel (but weaker) variant of copy-protection.

\section*{Acknowledgements}
Qipeng Liu is supported in part by the Simons Institute for the Theory of Computing, through a Quantum Postdoctoral Fellowship and by the DARPA SIEVE-VESPA grant Np.HR00112020023 and by the NSF QLCI program through grant number OMA-2016245. Any opinions, findings and conclusions or
recommendations expressed in this material are those of the author(s) and do not necessarily reflect the views of the United States
Government or DARPA.

Mark Zhandry is supported in part by an NSF CAREER award.
\qipeng{add more?}

\section{Preliminaries}
\label{sec:prelim}
\subsection{Basics}
We will briefly introduce some basic notations in our work and some preliminaries on quantum computing in this section.

\par We denote by $\secparam$ the security parameter. 
We write $\poly(\cdot)$ to denote an arbitrary polynomial and $\negl(\cdot)$ to denote an arbitrary negligible function. We say that an event happens with \textit{overwhelming probability} if the probability is at least $1 - \negl(\secparam)$.

Readers unfamiliar with quantum computation and quantum information could refer to~\cite{nielsen2002quantum} for a comprehensive introduction. 

\par Given Hilbert space $\Hs$, we write $\Ss(\Hs)$ for the unit sphere set $\{x :||x||_2=1\}$ in $\Hs$, $\mathcal{U}(\Hs)$ for the set of unitaries acting on Hilbert space $\Hs$, $\mathcal{D}(\Hs)$ for the set of density operators on $\Hs$. We write $\Hs_X$ to denote the Hilbert space associated with a quantum register $X$. Given two quantum states $\rho, \sigma$, we denote the (normalized) trace distance between them by \begin{align*}
    \tracedist{\rho}{\sigma} := \frac{1}{2}\left\|\rho - \sigma\right\|_{\mathsf{tr}}.
\end{align*} 

We say that two states $\rho, \sigma$ are \textit{$\delta$-close} if $\tracedist{\rho}{\sigma} \le \delta$.
\par A positive operator-valued measurement (POVM) on the Hilbert space $\mathcal{H}$ is defined as a set of positive semidefinite operators $\{E_i\}$ on $\Hs$ that satisfies $\sum_i E_i=I$. A projective measurement means the case that $E_i$s are projectors. 

A common technique in quantum computation is uncomputing~\cite{bennett1997strengths}. A quantum algorithm could be modeled as a unitary $U$ acting on some hilbert space $\Hs$, then perform measurement on output registers on without loss of generality. By uncomputation we mean that acting $U^\dag$ on the same hilbert space after the measurement. It is easy to examine that if the measurement outputs same result with overwhelming probability, the trace distance between the final state and the original state is negligible.

\paragraph{Quantum Oracle Algorithms}
A quantum oracle for a function $f$ is defined as the controlled unitary $O_f$: $O_f\ket{x}\ket{y}=\ket{x}\ket{y\oplus f(x)}$. We define a query to the quantum oracle as applying $O_f$ on the given quantum state once.

\par We say that a quantum adversary $\alice$ with access to oracle(s) is \textit{query-bounded} if it makes at most $p(\secparam)$ queries to each oracle for some polynomial $p(\cdot)$.

\subsection{Quantum Random Oracle Model (QROM)} This is the quantum analogue of Random Oracle Model, where we model a hash function $H$ as a random classical function, and it can be accessed by an adversary in superposition, modeled by the unitary $O_H$. \\

\par The following theorem, paraphrased from \cite{bennett1997strengths}, will be used for reprogramming oracles without adversarial detection on inputs which are not queried with large weight:

\begin{theorem}[\cite{bennett1997strengths}] \label{thm:bbbv}
 Let $\alice$ be an adversary with oracle access to $H: \bit^m \to \bit^n$ that makes at most $T$ queries. Define $\ket{\phi_i}$ as the global state after $\alice$ makes $i$ queries, and $W_y(\ket{\phi_i})$ as the sum of squared amplitudes in $\ket{\phi_i}$ of terms in which $\alice$ queries $H$ on input $y$. Let $\epsilon>0$ and let $F \subseteq [0,T-1] \times \bit^m$ be a set of time-string pairs such that $\sum_{(i,y) \in F} W_y(\ket{\phi_i}) \le \epsilon^2 / T $.
 \par Let $H'$ be an oracle obtained by reprogramming $H$ on inputs $(i,y) \in F$ to arbitrary outputs. Define $\ket{\phi_i'}$ as above for $H'$. Then, $\tracedist{\ket{\phi_T}}{\ket{\phi_{T}'}} \le \epsilon / 2$. \pnote{maybe we should use a different notation for trace distance? $T$ is used to denote both trace distance and the number of queries.. there is something strange here: fix $\epsilon$, let  $\tracedist{\ket{\phi_T}}{\ket{\phi_{T}'}} > \epsilon / 2$ and let $F=\emptyset$ then the theorem does not hold.} \fatih{Changed notation. If $F$ is empty, then the oracles are identical, so the trace distance is 0.}
\end{theorem}

Note that the theorem can be straightforwardly generalized to mixed states by convexity.

\subsection{More on Jordan's lemma}
\label{sec:Jordan}

We first recall the following version of Jordan's lemma, adapted from~\cite{Regev05} and ~\cite{Vidick21}:  %
\begin{lemma} \fatih{Generalized this to weighted averages.}
    Let $\weight \in [0,1]$, $\Hs$ be a finite-dimensional Hilbert space and let $\Pi_0, \Pi_1$ be any two projectors in $\Hs$, then there exists an orthogonal decomposition of $\Hs$ into one-dimensional and two dimensional subspaces $\Hs = \oplus_i \Ss_i$ that are invariant under both $\Pi_0$ and $\Pi_1$; each $\Ss_i$ is spanned by one or two eigenvectors of $\weight\Pi_0 + (1 - \weight)\Pi_1$. 
    
    Whenever $\Ss_i$ is $2$-dimensional, there is a basis for it in which $\Pi_0$ and $\Pi_1$ (restricting on $\Ss_i$) take the form:
    \begin{align*}
        \Pi_{0, \Ss_i} = \begin{pmatrix} 1 & 0 \\ 0 & 0 \end{pmatrix}  \quad\quad\text{ and }\quad\quad   \Pi_{1, \Ss_i} = \begin{pmatrix} c_i^2 & c_i s_i \\ c_i s_i & s_i^2 \end{pmatrix},
    \end{align*}
    where $c_i = \cos \theta_i$ and $s_i = \sin \theta_i$ for some principal angle $\theta_i \in [0, \pi/2]$. 
\end{lemma}

\begin{proof}
The proof for the case $\weight = 1/2$ can be found in the references above, and the generalization is straightforward. 
\end{proof}

We additionally show a relation between two eigenvalues in the same Jordan block. 
\begin{lemma} \label{lem:eigenvalues}
    For any two projectors $\Pi_0, \Pi_1$, let $\Ss_i$ be a $2$-dimensional subspace in the above decomposition. Let $\ket {\phi_0}, \ket{\phi_1}$ be two eigenvectors of $\weight\Pi_0 + (1 - \weight)\Pi_1$ that span $\Ss_i$ and $\lambda_0, \lambda_1$ be their eigenvalues. We have $\lambda_0 + \lambda_1 = 1$. 
\end{lemma}
\begin{proof}
    Restricting on $\Ss_i$, we have: 
    \begin{align*}
        \lambda_0 + \lambda_1 = \Tr\left[(\Pi_{0, \Ss_i} +\Pi_{1, \Ss_i})/2\right] = (1 + c_i^2 + s_i^2)/2 = 1. 
    \end{align*}
\end{proof}

\begin{corollary} \label{cor:orthogonal_eigen}
    For any two projectors $\Pi_0, \Pi_1$, let $\ket{\phi_0}$ and $\ket{\phi_1}$ be two eigenvectors of $\weight\Pi_0 + (1 - \weight)\Pi_1$ with eigenvalues $\lambda_0, \lambda_1$. If $\lambda_0 + \lambda_1 \ne 1$ and $\lambda_0 \ne \lambda_1$, then 
    \begin{align*}
        \langle \phi_0 |\Pi_0|\phi_1 \rangle = \langle \phi_0 |\Pi_1|\phi_1 \rangle = 0. 
    \end{align*}
\end{corollary}
\begin{proof}
    If $\lambda_0 + \lambda_1 \ne 1$, by \Cref{lem:eigenvalues}, $\ket {\phi_0}$ and $\ket {\phi_1}$ cannot be in the same Jordan block. Because $\ket{\phi_0}$ still belongs to the corresponding subspace $\Ss_0$ of its Jordan block after the action of $\Pi_0$, $\Pi_0 \ket {\phi_0}$ is orthogonal to $\ket {\phi_1}$. Similarly, $\Pi_1 \ket {\phi_0}$ is orthogonal to $\ket {\phi_1}$. 
\end{proof}

\fatih{ This particular fact is false for more than two projectors. Here is a counter-example:

\begin{align*}
    \Pi_1 = \ketbra{0}{0} \otimes I = \begin{pmatrix}
        1 & 0 & 0 & 0 \\
        0 & 1 & 0 & 0 \\
        0 & 0 & 0 & 0 \\
        0 & 0 & 0 & 0
    \end{pmatrix}
    , \quad \Pi_2 = I \otimes \ketbra{0}{0} = \begin{pmatrix}
        1 & 0 & 0 & 0 \\
        0 & 0 & 0 & 0 \\
        0 & 0 & 1 & 0 \\
        0 & 0 & 0 & 0
    \end{pmatrix}, \\
    \Pi_3 = \ketbra{+0}{+0} = \begin{pmatrix}
        1/2 & 0 & 1/2 & 0 \\
        0 & 0 & 0 & 0 \\
        1/2 & 0 & 1/2 & 0 \\
        0 & 0 & 0 & 0
    \end{pmatrix}, \quad \Pi = \frac{\Pi_1 + \Pi_2 + \Pi_3}{3} = \begin{pmatrix}
        5/6 & 0 & 1/6 & 0 \\
        0 & 1/3 & 0 & 0 \\
        1/6 & 0 & 1/2 & 0 \\
        0 & 0 & 0 & 0
    \end{pmatrix}.
\end{align*}

Two eigenvectors of $\Pi$ are given by \begin{align*}
    \ket{v_1} = \begin{pmatrix}
    1 + \sqrt{2} \\
    0 \\
    1 \\
    0
    \end{pmatrix}, \quad \ket{v_2} = \begin{pmatrix}
    1 - \sqrt{2} \\
    0 \\
    1 \\
    0
    \end{pmatrix},
\end{align*}
with eigenvalues $\lambda_1 = \frac{4 + \sqrt{2}}{6}$ and $\lambda_2 = \frac{4 - \sqrt{2}}{6}$. Observe that $\braket{v_1 | \Pi_1 | v_2} \ne 0$ and $\lambda_1 + \lambda_2 = 4/3 \ne 1$.

}

\subsection{Measuring Success Probability} \label{sec:measureprobability}
\fullversion{\subversion{\text{Here we discuss the background for the theorems in \Cref{sec:measureprobability}.}}
In this section, we give preliminaries on how to measure success probability of quantum programs (with respect to a test distribution). 
Part of this section is taken verbatim from \cite{aaronsonnew,coladangelo2021hidden}. Since this section will only be used for proving the strengthened monogamy-of-entanglement game of coset states in the quantum random oracle model (see \Cref{sec:moe_coset}), the reader can safely skip it to view our construction first, and return to this section when understanding the proof of the strengthened MOE game. %

\medskip

In classical cryptography, we are often interested in the success probability of a given program with respect to a test distribution. Assume that the test distribution is known to everyone and can be efficiently sampled, one can efficiently estimate the success probability of a given program within any inverse polynomial error. The estimating algorithm is fairly simple: just run the programs multiple times and output how many times the program succeeds. 
However, this method does not quite work when quantum programs are taken into account. One crucial reason is that the estimation algorithm only gets a single copy of the program. It is in general impossible to run the program multiple times without rewinding. However, rewinding a quantum program appears to be one of the difficulties in quantum cryptography. We refer the reader to \cite{z20} for a more in-depth discussion.

\paragraph{Measure Probability.}\  
In \cite{z20}, Zhandry formalizes a measurement operator for estimating the success probability of a quantum program. This operator is inefficient to implement, but Zhandry also shows how to efficiently estimate the probability with large statistical confidence in the same work (following the idea in QMA amplification \cite{marriott2005quantum}). We will discuss the efficient measurement procedure later in this section.

The starting point is that a binary POVM specifies the probability distribution over outcomes $\{0,1\}$ (``success'' or ``failure'') on any quantum program, but it does not uniquely determine the post-measurement state. Zhandry shows that, for any binary POVM $\cP = (P, I-P)$, there exists a nice projective measurement such that the post-measurement state is an eigenvector of $P$. In particular, Zhandry observes that there exists a projective measurement $\cE$ which \emph{measures} the success probability of a state with respect to $\cP$. More precisely,
\begin{itemize}
    \item $\cE$ outputs a  probability $p \in [0,1]$ from the set of eigenvalues of $P$. (We stress that $\cE$ actually outputs a real number $p$). %
    \item The post-measurement state upon obtaining outcome $p$ is an \emph{eigenvector} of $P$ with eigenvalue $p$; it is also an eigenvector of $Q = I-P$ with eigenvalue $1-p$. 
\end{itemize}

Note that since $\cE$ is projective, we are guaranteed that applying the same measurement again on the leftover state will yield the same outcome. Thus, what we obtain from applying $\cE$ is a state with a ``well-defined'' success probability with respect to $\cP$.

Furthermore, $\cE$ is compatible with $\cP$. In other words, one can safely measure the success probability of a program without disturbing the overall success probability.
We now give the formal theorem statement. 
}
\subversion{In this section we list theorems about simultaneously approximating the eigenvalues of a bipartite quantum program which are crucial tools in our security proofs.}
\begin{theorem}[Inefficient Measurement] \label{thm:inefficient_measure}
    Let $\cP = (P, Q)$ be a binary outcome POVM. 
    Let $\cD$ be the set of eigenvalues of $P$. 
    There exists a projective measurement $\cE = \{E_p\}_{p \in \cD}$ with index set $\cD$ that satisfies the following: for every quantum state $\rho$, let $\rho_p$ be the sub-normalized post-measurement state obtained after measuring $\rho$ with respect to $E_p$. That is, $\rho_p= E_p \rho E_p$. We have, 
    \begin{itemize}
        \item[(1)]  For every $p \in \cD$, $\rho_p$ is an eigenvector of $P$ with eigenvalue $p$;
        \item[(2)]  The probability of $\rho$ when measured with respect to $P$ is $\Tr[P \rho] = \sum_{p \in \cD} \Tr[P \rho_p]$.
    \end{itemize}
\end{theorem}

A measurement $\cE$ which satisfies these properties is the measurement in the common eigenbasis of $P$ and $Q=I-P$ (due to simultaneous diagonalization theorem, such common eigenbasis exists since $P$ and $Q$ commute). %
Let $P$ have eigenbasis $\{\ket {\psi_i}\}$ with eigenvalues $\{\lambda_i\}$. 
Without loss of generality, let us assume $\rho$ is a pure state $\ket{\psi}\bra{\psi}$ and $\{\lambda_i\}$ has no duplicated eigenvalues. %
We write $\ket{\psi}$ in the eigenbasis of $P$: $\ket{\psi} = \sum_{i} \alpha_i \ket{\psi_i}$. Applying $\cE$ will result in an outcome $\lambda_i$ and a leftover state $\ket {\psi_i}$ with probability $|\alpha_i|^2$.

Looking ahead, we will write a quantum program under the eigenbasis of $P$ in the proof of the strengthened MOE game. 

\medskip

\begin{theorem}[Inefficient Threshold Measurement] \label{thm:inefficient_threshold_measure}
    Let $\cP = (P, Q)$ be a binary outcome POVM. Let $P$ have eigenbasis $\{\ket {\psi_i}\}$ with eigenvalues $\{\lambda_i\}$. Then, for every $\gamma \in (0,1)$ there exists a projective measurement $\cE_{\gamma} = (E_{\leq \gamma}, E_{> \gamma})$ such that:
    \begin{itemize}
        \item[(1)] ${E}_{\leq \gamma}$ projects a quantum state into the subspace spanned by $\{\ket{\psi_i}\}$ whose eigenvalues $\lambda_i$ satisfy $\lambda_i \leq \gamma$; 
        \item[(2)] ${E}_{> \gamma}$ projects a quantum state into the subspace spanned by $\{\ket{\psi_i}\}$ whose eigenvalues $\lambda_i$ satisfy $\lambda_i > \gamma$.
    \end{itemize}
    
    Similarly, for every $\gamma \in (0, 1/2)$, there exists a projective measurement ${\cE'}_\gamma = (\widetilde{E}_{\leq \gamma}, \widetilde{E}_{> \gamma})$ such that:
    \begin{itemize}
        \item[(1)] $\widetilde{E}_{\leq \gamma}$ projects a quantum state into the subspace spanned by $\{\ket{\psi_i}\}$ whose eigenvalues $\lambda_i$ satisfy $|\lambda_i - \frac{1}{2}| \leq \gamma$; 
        \item[(2)] $\widetilde{E}_{> \gamma}$ projects a quantum state into the subspace spanned by $\{\ket{\psi_i}\}$ whose eigenvalues $\lambda_i$ satisfy $|\lambda_i - \frac{1}{2}| > \gamma$.
    \end{itemize}
\end{theorem}
It is easy to see how to construct $\cE_\gamma, \cE_\gamma'$ from $\cE$, e.g. by setting $\widetilde{E}_{\leq \gamma} = \sum_{i: |\lambda_i - 1/2| \leq \gamma} E_{\lambda_i}$. Note that for any quantum state $\rho$, $\Tr[\widetilde{E}_{> \gamma} \rho]$ is the weight over eigenvectors with eigenvalues $\lambda$ that are $\gamma$ away from $1/2$.

\fullversion{\paragraph{Efficient Measurement.}
The projective measurement $\cE$ above is not efficiently computable in general. 
However, they can be approximated if the POVM is a mixture of projective measurements, as shown by Zhandry \cite{z20}, using a technique first introduced by Marriott and Watrous \cite{marriott2005quantum}. 

Consider the following procedure as a binary POVM $\cP = (P, Q)$ acting on a quantum program $\rho$: samples a random challenge $r$, evaluates the program on $r$, and checks if the output is correct. 
This procedure can be viewed as (1). picking a uniformly random challenge $r$; (2). applying a \emph{projective measurement} $U_r$. In this case, $P = \frac{1}{R} \sum_r U_r$ where $R$ is the size of the challenge space. 
This POVM captures the situation where a challenger randomly samples a classical challenge and tests if a quantum program's classical outcome is correct on that challenge. 

}

Below, we give the formal theorem statement about efficient approximated threshold measurement, which is adapted from Theorem 6.2 in \cite{z20} and Lemma 3 in \cite{aaronsonnew}.

\begin{theorem}[Efficient Threshold Measurement] \label{thm:thres_approx_asymmetric}
    Let $\cP_b = (P_b,Q_b)$ be a binary outcome POVM over Hilbert space $\Hs_b$ that is a mixture of projective measurements for $b \in \{1,2\}$. Let $P_b$ have eigenbasis $\{\ket {\psi_i^b}\}$ with eigenvalues $\{\lambda_i^b\}$. For every $\gamma_1, \gamma_2 \in (0, 1), 0 < \epsilon < \min(\gamma_1/2, \gamma_2/2, 1 - \gamma_1, 1 - \gamma_2)$ and $\delta > 0$, there exist efficient binary-outcome quantum algorithms, interpreted as the POVM element corresponding to outcome 1, $\ati_{\cP_b, \gamma}^{\epsilon, \delta}$ such that for every quantum program $\rho \in \density{1} \otimes \density{2}$ the following are true about the product algorithm $\ati_{\cP_1, \gamma_1}^{\epsilon, \delta} \otimes \ati_{\cP_2, \gamma_2}^{\epsilon, \delta}$:
    \begin{itemize}
        \item[(0)] Let $ (E^b_{\leq \gamma}, E^b_{> \gamma})$ be the inefficient threshold measurement in \Cref{thm:inefficient_threshold_measure} for $\Hs_b$. 
        \item[(1)] The probability of measuring 1 on both registers satisfies $$\Tr\bracketsS{\brackets{ \ati_{\cP_1, \gamma_1}^{\epsilon, \delta} \otimes \ati_{\cP_2, \gamma_2}^{\epsilon, \delta} } \rho } \ge \Tr\bracketsS{ \brackets{ E^1_{> \gamma_1 + \epsilon} \otimes E^2_{> \gamma_2 + \epsilon} } \cdot \rho} - 2\delta.$$
        \item[(2)] The post-measurement state $\rho'$ after getting outcome (1,1) is $4\delta$-close to a state in the support of $\bracketsC{\ket{\psi^1_i}\ket{\psi^2_j}}$ such that $\lambda_i^1 > \gamma_1 - 2\epsilon$ and $\lambda^2_j > \gamma_2 - 2\epsilon$. 
        
        \item[(3)] The running time of the algorithm is polynomial in the running time of $P_1, P_2$, ${1}/{\epsilon}$ and $\log(1/\delta)$. 
    \end{itemize}
\end{theorem}

Intuitively the theorem says that if a quantum state $\rho$ has weight $p$ on eigenvectors of $(P_1, P_2)$ with eigenvalues greater than $(\gamma_1 + \epsilon,\gamma_2 + \epsilon)$\pnote{should this be $(\gamma_1 - \epsilon,\gamma_2-\epsilon)$?}\fatih{I fixed the statement and adjusted the proofs accordingly.}, then the quantum algorithm will produce (with probability at least $p - 2\delta$) a post-measurement state which has weight $1 - 4\delta$ on eigenvectors with eigenvalues greater than $(\gamma_1 - 2\epsilon, \gamma_2 - 2\epsilon)$.

\par In this paper, we will work with indistinguishability games.
Therefore, we will particularly be interested in the projective measurement that projects onto eigenvectors with eigenvalues away from $1/2$ (meaning its behavior is more than random guessing). For this reason, we will need the following symmetric version of \Cref{thm:thres_approx_asymmetric}:

\begin{theorem}[Efficient Symmetric Threshold Measurement] \label{thm:thres_approx}
    Let $\cP_b = (P_b,Q_b)$ be a binary outcome POVM over Hilbert space $\Hs_b$ that is a mixture of projective measurements for $b \in \{1,2\}$. Let $P_b$ have eigenbasis $\{\ket {\psi_i^b}\}$ with eigenvalues $\{\lambda_i^b\}$. For every $\gamma_1, \gamma_2 \in (0, 1/2), 0 < \epsilon < \min(\gamma_1/2,\gamma_2/2)$, and $ \delta > 0$, there exist efficient binary-outcome quantum algorithms, interpreted as the POVM element corresponding to outcome 1, $\sati_{\cP_b, \gamma}^{\epsilon, \delta}$ %
    such that for every quantum program $\rho \in \density{1} \otimes \density{2}$ the following are true about the product algorithm $\sati_{\cP_1, \gamma_1}^{\epsilon, \delta} \otimes \sati_{\cP_2, \gamma_2}^{\epsilon, \delta}$:
    \begin{itemize}
        \item[(0)] Let $ (\widetilde{E}^b_{\leq \gamma_b}, \widetilde{E}^b_{> \gamma_b})$ be the inefficient threshold measurement in \Cref{thm:inefficient_threshold_measure} for $\Hs_b$. 
        \item[(1)] The probability of measuring 1 on both registers satisfies $$\Tr\bracketsS{\brackets{ \sati_{\cP_1, \gamma_1}^{\epsilon, \delta} \otimes \sati_{\cP_2, \gamma_2}^{\epsilon, \delta} } \rho } \ge \Tr\bracketsS{ \brackets{ \widetilde{E}^1_{> \gamma_1 + \epsilon} \otimes \widetilde{E}^2_{> \gamma_2 + \epsilon} } \cdot \rho} - 2\delta.$$

        \item[(2)] The post-measurement state $\rho'$ after getting outcome (1,1) is $4\delta$-close to a state in the support of $\bracketsC{\ket{\psi^1_i}\ket{\psi^2_j}}$ such that $|\lambda_i^1 - 1/2| > \gamma_1 - 2\epsilon$ and $|\lambda^2_j - 1/2| > \gamma_2 - 2\epsilon$.

        \item[(3)] The running time of the algorithm is polynomial in the running time of $P_1, P_2$, ${1}/{\epsilon}$ and $\log(1/\delta)$. 
    \end{itemize}
\end{theorem}

\noindent %

\subsection{Unclonable Encryption}

In this subsection, we provide the definition of unclonable encryption schemes. 
By unclonable encryption, we are refering to the security defined in ~\cite{ananth2021uncloneable}. This is a variant of the original security definition in \cite{broadbentUncloneableQuantumEncryption2020}, which forces one of $m_0,m_1$ to be uniformly random. We would remark that our security is stronger than the original one in \cite{broadbentUncloneableQuantumEncryption2020}, since in our definition $m_0,m_1$ can be arbitrarily chosen. %

\begin{definition}
\label{def:Unclonable_enc}
An unclonable encryption scheme is a triple of efficient quantum algorithms $(\gen,\enc,\dec)$ with the following interface:
\begin{itemize}
    \item $\gen(1^\lambda):\sk$ on input a security parameter $1^\lambda$, returns a classical key $\sk$.%
    \item $\enc(\sk,\ket{m}\bra{m}):\rho_{ct}$ takes the key $\sk$ and the message $\ket{m}\bra{m}$ for $m \in \{0,1\}^{\mathrm{poly}(\lambda)}$, outputs a quantum ciphertext $\rho_{ct}$. 
    \item $\dec(\sk,\rho_{ct}):\rho_m$ takes the key $\sk$ and the quantum ciphertext $\rho_{ct}$, outputs a message in the form of quantum states $\rho_m$.
\end{itemize}

\paragraph{Correctness.} The following must hold for the encryption scheme. For $\sk\leftarrow \gen(1^{\lambda})$, we must have $\Tr[\ket{m}\bra{m}\dec(\sk,\enc(\sk,\ket{m}\bra{m}))]\geq 1-\negl(\lambda)$.

\paragraph{Unclonability.} In the following sections, we focus on unclonable IND-CPA security. To define our unclonable security, we introduce the following security game.

\end{definition}

\begin{definition}[Unclonable IND-CPA game]
Let $\lambda\in \mathbb{N}^+$. Given encryption scheme $\mathcal{S}$, consider the following game against the adversary $(\mathcal{A},\mathcal{B},\mathcal{C})$.
\begin{itemize}
    \item The adversary $\As$ generates $m_0,m_1\in\{0,1\}^{n(\lambda)}$ and sends to the challenger as the chosen plaintext.
    \item The challenger randomly chooses a bit $b\in\{0,1\}$ and returns $\enc(\sk,m_b)$ to $\As$. $\As$ produces a quantum state $\rho_{BC}$ in register $B$ and $C$, and sends corresponding registers to $\mathcal{B}$ and $\mathcal{C}$.
    \item $\Bs$ and $\Cs$ receive the key $\sk$, and output bits $b_{\Bs}$ and $b_{\Cs}$ respectively
\end{itemize}
 and the adversary wins if $b_\Bs=b_\Cs=b$. 

\end{definition}

We denote the advantage (success probability) of above game by $\adv_{\mathcal{G},\As,\Bs,\Cs}(\lambda)$. We say that scheme $\mathcal{S}$ is informational (computational) secure if for all(efficient) adversaries $(\mathcal{G},\mathcal{A},\mathcal{B},\mathcal{C})$,
\begin{align*}
    \adv_{\mathcal{G},\As,\Bs,\Cs}(\lambda)\leq \frac{1}{2}+\negl(\lambda).
\end{align*}

\fullversion{\section{On the Impossibility of Deterministic Schemes}
\label{sec:imp}
In this section, we provide an impossibility result for \emph{deterministic information-theoretically secure} schemes. 
This result suggests that either computational assumptions or randomness is necessary for achieving unclonable encryption with optimal security. 
We also noticed that previously  in~\cite{majenzLimitationsUncloneableEncryption2021}, the authors have provided an impossibility result for more general schemes. Nevertheless, our result provides a better asymptotic lower bound for deterministic schemes and is based on observations on Haar random states.

To be precise, we define deterministic schemes as follows:
\begin{definition}[Deterministic Scheme]
We call an encryption scheme $(\gen,\enc,\dec)$ is a deterministic encryption scheme if it satisfies following:
\begin{itemize}
    \item The encryption algorithm $\enc$ can be realized as a unitary $U_{\sk}$ acting on the plaintext register $\ket{m}$ and ancillary bits initialized to $0$, resulting in the ciphertext pure state in the form $\ket{c_{\sk}}$ of length $\lambda$. %
    \item The decryption algorithm $\dec$ acts the inverse $U_{\sk}^{\dag}$ on received registers, then measures in computational basis to obtain the message.
\end{itemize}
\end{definition}
The correctness of deterministic schemes is satisfied. An example of deterministic scheme is the following: let $\sk$ encode two (arbitrary and) orthogonal states $\ket {\phi_0}, \ket{\phi_1}$; a message $b$ is mapped to $\ket{\phi_b}$. Another example is the conjugate encryption defined in \cite{broadbentUncloneableQuantumEncryption2020}:
\begin{enumerate}
    \item $\sk=(r,\theta)$ where $r,\theta$ is independent random samples from $\{0,1\}^n$
    \item $\enc(\sk,m)=\ket{(m\oplus r)^\theta}\bra{(m\oplus r)^\theta}$, where $\ket{x^\theta}=H^{\theta_1}\otimes H^{\theta_2}\otimes\dots \otimes H^{\theta_n}\ket{x}$ is the BB84 state.
    \item $\dec(\sk,\rho)$: computes $\rho'=H^\theta\rho H^\theta$, measures $\rho'$ in computational basis to obtain $c$, obtaining $m=c\oplus r$.
\end{enumerate}

Though the authors in ~\cite{broadbentUncloneableQuantumEncryption2020} 
have already proven this scheme does not satisfy the unclonable IND-CPA security, our attack scheme provided a no go theorem for a larger class of possible constructions.

For these schemes, we provide a universal adversary for the unclonable IND-CPA game.

\begin{theorem}
\label{thm:IT_imp}
For any deterministic encryption scheme, we have a universal information-theoretical adversary $(\mathcal{G},\As,\Bs,\Cs)$ that satisfies 
\begin{align*}
    \adv_{\mathcal{G},\As,\Bs,\Cs}(\lambda)\geq 0.568,
\end{align*}
as $\lambda\to\infty$.
\end{theorem}

Since any deterministic encryption scheme can only suffice one-time security, we also considered whether our result can be extended to general encryption schemes that take randomness as input, such as the following scheme inspired by~\cite{GL89}. 
\begin{itemize}
    \item $\sk=(\theta,u)$ for $\theta,u\leftarrow\{0,1\}^\lambda$.
    \item $Enc_k(m,r)=\ket{r^{\theta}}\ket{\langle r,u\rangle\oplus m}$ for $m\in\{0,1\}$, $r\leftarrow\{0,1\}^\lambda$.
    \item $Dec_k(\rho):$ Decode $r$ by applying $H^\theta$ on first $\lambda$ register, and measure $\rho$ in computational basis to get $ct$. We can extract $m=\langle ct_{1\dots\lambda}, u\rangle\oplus ct_{\lambda+1}$.
\end{itemize}
However, our impossibility result met some barriers in the generalization. We would try to characterize them as following:
\begin{itemize}
    \item Since in quantum algorithms, randomness is generated intrinsically from measurements. Consider implementing a classical randomized algorithm by quantum circuits, the random bits in the classical algorithm would be replaced by measuring $\ket{+}$ states in the computational basis. Thus for general encryption algorithms, they should be modeled as quantum channels rather than unitaries, with cipher texts modeled as mixed states accordingly. However, the understanding on the actions of random unitaries on mixed states is a much less studied and more complicated problem.
    \item Our adversary $(\As,\Bs,\Cs)$ also relies on \emph{all} information of the cipher text states to decide its measurement. But if the encryption algorithm additionally takes some randomness, then the adversary $\Bs$ and $\Cs$ cannot decide the actual ciphertext state. 
\end{itemize}

\subsection{Preliminaries on Haar Measure}
To prove our result, we provide a quick introduction to the theorems related to Haar measure in this subsection. For more information on Haar measure, readers can refer to ~\cite{watrousTheoryQuantumInformation2018}. We denote the uniform spherical measure on unit sphere $\Ss((\Cp^2)^{\otimes n})$ as $\mu_n$, the Haar measure on the unitary group $\mathcal{U}((\Cp^2)^{\otimes n})$ as $\eta_n$.

The following lemma relates the Haar measure on unitary operators to uniform spherical measure.

\begin{lemma}
\label{lem:mueta}
Let $f$ be a function from $\Ss((\Cp^2)^{\otimes n})\times\Ss((\Cp^2)^{\otimes n})\rightarrow\mathbb{R}$. Then for any two fixed vectors $\ket{\phi_0},\ket{\phi_1}\in \Ss((\Cp^2)^{\otimes n})$ such that $\braket{\phi_0|\phi_1}=0$, we have that 
\begin{align*}
    \E_{U\leftarrow \eta_n} f(U\ket{\phi_0},U\ket{\phi_1})=\E_{\ket{\psi_0},\ket{\psi_1}\leftarrow \mu_n,\braket{\psi_0|\psi_1}=0}f(\ket{\psi_0},\ket{\psi_1}).
\end{align*}

\end{lemma}

We introduce Lévy's lemma, which could be viewed as the counterpart of Chernoff bound on the uniform spherical measure.

\begin{lemma}[Lévy's Lemma]
\label{lem:levy}
Let $f$ be a function from $\Ss((\Cp^2)^{\otimes n})\to\mathbb{R} $ that satisfies
\begin{align*}
    |f(\ket{\phi})-f(\ket{\psi})|\leq \kappa||\ket{\phi}-\ket{\psi}||_2,
\end{align*} 
for some $\kappa>0$. Then there exists a universal $\delta>0$ for which the following holds.  For every $\epsilon>0$:
\begin{align*}
    \Pr_{\ket{\psi}\leftarrow\mu_n}\left[\left|f(\ket{\psi})-\E_{\ket{\phi}\leftarrow\mu_n}[f(\ket{\phi})]\right|\geq\epsilon\right]\leq3\exp\left(-\frac{\delta\epsilon2^n}{\kappa^2}\right).
\end{align*}
\end{lemma}

The following simplified theorem from \cite{MZB16} plays a crucial role in our proof. 

\begin{theorem}
\label{thm:limit_dis}
Let $\ket{\psi_1},\ket{\psi_2}\in \Ss((\Cp^2)^{\otimes 2 n})$ be two states independently sampled from $\mu_{2n}$. Then let $\rho_1,\rho_2$ be the corresponding reduced density matrix in the first $n$ qubit register. As $n\to \infty$, the trace distance $\tracedist{\rho_1}{\rho_2}$ almost surely converges to
\begin{align*}
    \tracedist{\rho_1}{\rho_2}\stackrel{a.s.}{\longrightarrow}\frac{1}{4}+\frac{1}{\pi}\approx 0.568.
\end{align*}
\end{theorem}

For simplicity, in this section, $\E_{\ket{\psi}}$ stands for taking expectation over $\ket{\psi}$ sampled from uniform spherical measure on corresponding Hilbert space, $\E_V$ stands for $V$ over Haar measure respectively.

\subsection{Attack schemes}

We are ready to present an attack for any deterministic information-theoretically secure schemes. 

\paragraph{Attack.}
\begin{itemize}
    \item For the adversary $\As$, it first chooses $00\dots00,00\dots01$ and sends to the challenger. After receiving the $n$ qubit ciphertext state $\ket{ct_k}$, it applies a random Haar unitary $V$, then divides the output register into two parts, $R_{\Bs}$ for the qubits indexed $[1,\frac{\lambda}{2}]$, $R_{\Cs}$ for the qubits indexed $[\frac{\lambda}{2}+1,\lambda]$.
    \item $\As$ then sends two registers respectively to $\Bs$ and $\Cs$, together with the description of $V$
    \footnote{Here we actually mean sending the corresponding minimal distance $\widetilde{V}$ in the $\epsilon$-net of $\mathrm{U}(\Hs_2^{\otimes\lambda})$ to approximate the distribution of $V$. Thus we can sample from a finite set instead. %
    Since we have a constant advantage in the end, we can take $\epsilon$ small enough such that it will only have a negligible effect on our result. }.
    \item   With the given information, $\Bs$ and $\Cs$ can perform POVMs $\{\Pi_b^B\}_b$ and $\{\Pi_b^C\}_b$ to distinguish different messages. We will define the POVMs in detail in the following section.
\end{itemize}
The success probability of our attack scheme is equal to the success probability of the following game.  

\begin{definition}
Let $\lambda\in \mathbb{N}^+$. Consider the following game with a challenger and an (unbounded) adversary $(\Bs,\Cs)$. 
\begin{itemize}
    \item The challenger generates two Haar random states $\ket{\phi_0},\ket{\phi_1}$ with restriction $\braket{\phi_0|\phi_1}=0$ and sends the description of two states\footnote{Similarly, $\As$ sends an element in the $\epsilon$-net of $\mathcal{S}(\Hs_2^{\otimes \lambda})$ in implementation.} to $\Bs$ and $\Cs$.
    \item The challenger randomly chooses $b_{ch}\in\{0,1\}$, and divides the state $\ket{\phi_{b_{ch}}}$ into two parts, $R_{\Bs}$ for the qubits indexed $[1,\frac{\lambda}{2}]$, $R_{\Cs}$ for the qubits indexed $[\frac{\lambda}{2}+1,\lambda]$.
    \item $\Bs$ and $\Cs$ perform POVM $\{\Pi_{b_B}^B\}_{b_B}$ and $\{\Pi_{b_C}^C\}_{b_C}$ on their received registers, outputs $b_B$ and $b_C$ from measurement results.
\end{itemize}
The adversary wins the game if $b_B=b_C=b_{ch}$.
\end{definition}

 The success probability of our distinguishing game is given by the following optimization problem: 
\begin{align*}
    \max_{\Pi_0^B,\Pi_1^B,\Pi_0^C,\Pi_1^C}&\frac{1}{2}\left(\bra{\phi_0}\Pi_0^B\otimes\Pi_0^C\ket{\phi_0}+\bra{\phi_1}\Pi_1^B\otimes\Pi_1^C\ket{\phi_1}\right)\\
    \textit{s.t.}\ \ &\Pi_0^B+\Pi_1^B=I_{\frac{\lambda}{2}}
                    ,\Pi_0^C+\Pi_1^C=I_{\frac{\lambda}{2}},\\
                    &0\leq\Pi_i^B\leq I_{\frac{\lambda}{2}},
                    0\leq\Pi_i^C\leq I_{\frac{\lambda}{2}}.
\end{align*}

The $\frac{1}{2}$ comes from the requirement that the challenger sends $\ket{\phi_0},\ket{\phi_1}$ with equal probability. We denote this probability as $G(\ket{\phi_0},\ket{\phi_1})$. For simplicity, in following sections we will abbreviate $\{\Pi_{b_B}^B\}_{b_B}$ and $\{\Pi_{b_C}^C\}_{b_C}$ as $\{\Pi^B\}$ and $\{\Pi^C\}$ respectively.

In our attack scheme, our success probability is given by $\E_k\E_V[G(VU_k\ket{0\dots00},VU_k\ket{0\dots01})]$. By lemma \ref{lem:mueta}, we have that

\begin{align*}
    \E_V\left[G\left(VU_k\ket{0\dots00},VU_k\ket{0\dots01}\right)\right]=\E_{\ket{\phi_0},\ket{\phi_1}\braket{\phi_0|\phi_1}=0}\left[G\left(\ket{\phi_0},\ket{\phi_1}\right)\right]=\Pr[\text{$(\Bs,\Cs)$ wins} ]
\end{align*}

Then we can provide a lower bound for the success probability via following inequalities:
\begin{align*}
    &\, \Pr[\text{$(\Bs,\Cs)$ wins}|\ket{\phi_0},\ket{\phi_1}] \\
    =&\, \max_{\{\Pi^B\},\{\Pi^C\}}\Pr[(b_B=b_{ch})\wedge(b_C=b_{ch})|\{\Pi^B\},\{\Pi^C\},\ket{\phi_0},\ket{\phi_1}]\\
    \geq&\, 1-\min_{\{\Pi^B\}}\Pr[b_B\neq b_{ch}|\{\Pi^B\},\ket{\phi_0},\ket{\phi_1}]-\min_{\{\Pi^C\}}\Pr[b_C\neq b_{ch}|\{\Pi^C\},\ket{\phi_0},\ket{\phi_1}]\\
    =&\, 1-\frac{1}{2}(1-\tracedist{\rho_0^B}{\rho_1^B})-\frac{1}{2}(1-\tracedist{\rho_0^C}{\rho_1^C})\\
    =&\, \frac{1}{2}(\tracedist{\rho_0^B}{\rho_1^B}+\tracedist{\rho_0^C}{\rho_1^C}),
\end{align*}
where the first line is by definition, the second line follows from union bound, the third line is by the property of trace distance.

Then by taking expectation, we have that for large enough $\lambda$,
\begin{align*}
    \Pr[\text{$(\Bs,\Cs)$ wins}]&=\E_{\ket{\phi_0},\ket{\phi_1},\braket{\phi_0|\phi_1}=0}\Pr[\text{$(\Bs,\Cs)$ wins}|\ket{\phi_0},\ket{\phi_1}]\\
    &\geq\E_{\ket{\phi_0},\ket{\phi_1},\braket{\phi_0|\phi_1}=0}\left[\frac{1}{2}(\tracedist{\rho_0^B}{\rho_1^B}+\tracedist{\rho_0^C}{\rho_1^C})\right]\\
    &\geq \E_{\ket{\phi_0},\ket{\phi_1}}\left[\frac{1}{2}(\tracedist{\rho_0^B}{\rho_1^B}+\tracedist{\rho_0^C}{\rho_1^C})\right]-\negl(\lambda)\\
    &\geq \frac{1}{4}+\frac{1}{\pi}-\epsilon\geq 0.568,
\end{align*}
where the first line is by definition, the second line is by the inequality before, the third line is by concentration property of the Haar measure, the last line is by theorem \ref{thm:limit_dis} as $\lambda \to \infty$. Thus we finished the proof of theorem \ref{thm:IT_imp}

Here we provide rigorous proof of the third line. Note that for an arbitrary $\ket{\phi_1}$, given $\ket{\phi_0}$ it can be written as $\ket{\phi_1}=a\ket{\phi_0}+\sqrt{1-|a|^2}\ket{\phi_0^{\bot}}$, where $a=\braket{\phi_0|\phi_1}$ and $\braket{\phi_0|\phi_0^{\bot}}=0$. By symmetry, we have that $\E_{\ket{\phi_1}}[|a|^2]=\frac{1}{2^{\lambda}}$. Taking $\epsilon={\lambda}{2^{-\frac{\lambda}{2}}},\kappa=2$ in lemma \ref{lem:levy}, we obtain that 
\begin{align*}
    \Pr\left[\left||a|^2-\frac{1}{2^{\lambda}}\right|\geq \frac{\lambda}{2^\frac{\lambda}{2}}\right]\leq 3\exp\left(-\frac{\delta\lambda^2}{4}\right),
\end{align*}
thus we can derive that 
\begin{align*}
    \E_{\ket{\phi_1}}[|a|]\leq 3\exp\left(-\frac{\delta\lambda^2}{4}\right) \cdot 1 + 1 \cdot \frac{\sqrt{\lambda}+1}{2^{\frac{\lambda}{4}}}=\negl(\lambda).
\end{align*}

Consider the trace distance $\tracedist{\rho_0^B}{\rho_1^B}$ for two random states $\ket{\phi_0},\ket{\phi_1}$. By definition it can be rewritten as $|\Tr_C[\ket{\phi_0}\bra{\phi_0}-\ket{\phi_1}\bra{\phi_1}]|_1$, then following the decomposition of $\ket{\phi_1}=a\ket{\phi_0}+\sqrt{1-|a|^2}\ket{\phi_0^{\bot}}$, we expand the expectation of the term as

\begin{align*}
    &\E_{\substack{a \\ \ket{\phi_0},\ket{\phi_0^\bot} \\ \braket{\phi_0|\phi_0^\bot}=0}}\left[\frac{1}{2}\left|\Tr_C[(1-|a|^2)(\ket{\phi_0}\bra{\phi_0}-\ket{\phi_0^\bot}\bra{\phi_0^\bot})-\sqrt{1-|a|^2}(a\ket{\phi_0}\bra{\phi_0^\bot}+a^*\ket{\phi_0^\bot}\bra{\phi_0})]\right|_1\right]\\
    \leq& \E_{\substack{\ket{\phi_0},\ket{\phi_1}\\ \braket{\phi_0|\phi_1}=0}}\left[\frac{1}{2}\left|\Tr_C[\ket{\phi_0}\bra{\phi_0}-\ket{\phi_1}\bra{\phi_1}]\right|_1\right]+\E_a[|a|]\\
    \leq&\E_{\substack{\ket{\phi_0},\ket{\phi_1}\\ \braket{\phi_0|\phi_1}=0}}\left[\tracedist{\rho^B_0}{\rho^B_1}\right]+\negl(\lambda),
\end{align*}
where the second line is by definition, the third line is from the decomposition of $\ket{\phi_1}$, the fourth line is by triangle inequality and renaming $\ket{\phi_0^\bot}$ to $\ket{\phi_1}$, the last line is by definition and previous bounds on $\E[|a|]$.
}

\section{More on Coset States} \label{sec:moe_coset}

In this section, we will recall the basic properties of coset states. We will then introduce a strengthened \MOE game in the quantum random oracle model (QROM), upon which we will build our unclonable encryption scheme. The last subsection is devoted to prove the security of this strengthened game. 

\subsection{Preliminaries}

In this subsection, we recall the basic definitions and properties of coset states in \cite{coladangelo2021hidden}. 
Let $A \subseteq \F_2^n$ be a subspace. Define its orthogonal complement of $A$ as $A^\perp = \{ b \in \mathbb{F}_2^n \,|\,  \langle a, b\rangle \bmod 2 = 0 \,,\, \forall a \in A \}$. It satisfies $\dim(A) + \dim(A^\perp) = n$. We also let $|A| = 2^{\dim(A)}$ denote the size of $A$.

\begin{definition}[Coset States]
For any subspace $A \subseteq \mathbb{F}_2^n$ and vectors $s, s' \in \mathbb{F}_2^n$, the coset state $\ket {A_{s,s'}}$ is defined as:
\begin{align*}
    \ket {A_{s,s'}} = \frac{1}{\sqrt{|A|}} \sum_{a \in A} (-1)^{\langle s', a\rangle} \ket {a + s}\,.
\end{align*}
\end{definition}

By applying $H^{\otimes n}$ to the state $\ket {A_{s,s'}}$, one obtains exactly $\ket {A^{\perp}_{s', s}}$. Given $A, s, s'$, the coset state is efficiently constructible. 

For a subspace $A$ and vectors $s,s'$, we define $A+s = \{v +s : v \in A\}$, and $A^{\perp}+s' = \{v +s': v \in A^{\perp}\}$. We define $P_{A+s}$ and $P_{A^\perp+s'}$ as the membership checking oracle for both cosets.

It is also convenient for later sections to define a canonical representation of a coset $A+s$, with respect to subspace $A$,

\begin{definition}[Canonical Representative of a Coset]
\label{def:canonical_vec_func}
    For a subspace $A$, we define the function $\can_A(\cdot)$ such that $\can_A(s)$ is the lexicographically smallest vector contained in $A + s$. We call this the canonical representative of coset $A+s$.
\end{definition}

If $\tilde{s} \in A + s$, then $\can_A(s) = \can_A(\tilde{s})$.
We also note that $\can_A(\cdot)$ is polynomial-time computable given the description of $A$. Accordingly, we can efficiently sample from $\canonicalset(A) := \bracketsC{\can_A(s): s \in \F_2^n}$, which denotes the set of canonical representatives for $A$.\\

\par For a fixed subspace $A$, the coset states $\bracketsC{\ket{A_{s,s'}}}_{s \in \canonicalset(A), s' \in \canonicalset(A^\perp)}$ form an orthonormal basis.
{(See Lemma C.2 in \cite{coladangelo2021hidden})}

\medskip
Next, we recall the regular direct product and MOE properties of coset states. These properties will be used to prove the strengthened \MOE property. 

\paragraph{Direct Product Hardness}

\begin{theorem}[Theorem 4.5,4.6 in \cite{coladangelo2021hidden}]
\label{thm:dp_hard}
Let $A\subseteq \mathbb{F}_2^{\lambda}$ be a uniformly random subspace of dimension $\frac{\lambda}{2}$, and $s,s'$ be two uniformly random vectors from $\mathbb{F}_2^\lambda$. Let $\epsilon>0$ such that $1/\epsilon=o(2^{n/2})$. Given one copy of $\ket{A_{s,s'}}$ and oracle access to $P_{A+s}$ and $P_{A^\perp + s'}$, an adversary needs $\Omega(\sqrt{\epsilon}2^{\lambda/2})$ queries to output a pair $(v,w)$ that $v\in A+s$ and $w\in A^\perp+s'$ with probability at least $\epsilon$.
\end{theorem}

An important corollary immediately follows.

\begin{corollary}
\label{col:dp_negl}
There exists an exponential function $\exp$ such that, for any query-bounded (polynomially many queries to $P_{A+s}, P_{A^\perp + s'}$) adversary, its probability to output a pair $(v,w)$ that $v\in A+s$ and $w\in A^\perp+s'$ is smaller than $1/{\exp{(\lambda)}}$.
\end{corollary}

\paragraph{Monogamy-of-Entanglement (with Membership Checking Oracles).}

\begin{definition}
\label{def:regMOE}
Let $\lambda \in \mathbb{N}^+$. Consider the following game between a challenger and an adversary $(\As, \Bs, \Cs)$.
\begin{itemize}
    \item The challenger picks a uniformly random subspace $A \subseteq \mathbb{F}_2^\lambda$ of dimension $\frac{\lambda}{2}$, and uniformly random vectors $(s, s') \in \canonicalset(A) \times \canonicalset(A^\perp)$.
    It sends $\ket{A_{s,s'}}$ to $\As$. 
    \item $\As, \Bs, \Cs$ get (quantum) oracle access to $P_{A+s}$ and $P_{A^\perp + s'}$. 
    \item $\As$ creates a bipartite state on registers $\mathsf{B}$ and $\mathsf{C}$. Then, $\As$ sends register $\mathsf{B}$ to $\Bs$, and $\mathsf{C}$ to $\Cs$. 
    \item The description of $A$ is then sent to both $\Bs, \Cs$. 
    \item $\Bs$ and $\Cs$ return respectively $(s_1,s_1')$ and $(s_2, s_2')$.
\end{itemize}
$(\As, \Bs, \Cs)$ wins if and only if for $i \in \{1,2\}$, $s_i = s$ and  $s_i' = s'$.
\end{definition}

\medskip

We denote the advantage (success probability) of the above game by $\adv_{\As, \Bs, \Cs}(\lambda)$. 
We have the following theorem. 
\begin{theorem}[Theorem 4.14, 4.15 in \cite{coladangelo2021hidden}]
\label{thm: monogamy info}
There exists an exponential function $\exp$ such that, for every $\lambda \in \mathbb{N}^+$, for any query-bounded (polynomially many queries to $P_{A+s}, P_{A^\perp + s'}$) adversary $(\As, \Bs, \Cs)$, 
$$\adv_{\As, \Bs, \Cs}(\lambda) \leq 1/\exp(\lambda)\,.$$
\end{theorem}
Note that in \cite{coladangelo2021hidden}, the authors only proved the above theorem for a sub-exponential function and membership checking oracles are given in the form of indistinguishability obfuscation (iO). The proof trivially holds if we replace iO with VBB obfuscation (quantum access to these oracles). Culf and Vidick \cite{culf2021monogamy} further proved the theorem holds for an exponential function.

\subsection{Strengthened MOE Game in the QROM}

In this subsection, we will introduce the strengthened MOE game in the QROM and state our main theorem. We present the proof in the next section.

\begin{definition}
Let $\lambda \in \mathbb{N}^+$. Consider the following security game between a challenger and an adversary $(\As, \Bs, \Cs)$ with a random oracle $H:\F_2^\lambda \times \F_2^\lambda \to \{0,1\}^{n(\lambda)}$ .
\begin{itemize}
    \item The adversary $\As$ generates $\Delta\in\{0,1\}^{n(\lambda)}$ and sends $\Delta$ to the challenger.
    \item The challenger samples a random subspace $A \subseteq \F_2^\lambda$ of dimension $\lambda/2$ and two random vectors $(s,s') \in \canonicalset(A) \times \canonicalset(A^\perp)$. %
    The challenger also randomly chooses a bit $b\in\{0,1\}$ and calculates $w=H(s,s')\oplus( b\cdot\Delta)$.
    
    It gives $\ket{A_{s, s'}}$ and $w$ to $\As$.
    
    \item $\As, \Bs, \Cs$ get (quantum) oracle access to $P_{A+s}$ and $P_{A^\perp + s'}$. %
    
    \item $\As$ produces a quantum state over registers $\B \C$ and sends $\B$ to $\Bs$ and $\C$ to $\Cs$.
    
    \item $\Bs, \Cs$ are given the description of $A$, they try to produce bits $b_{\Bs}, b_{\Cs}$. 
\end{itemize}
$(\As, \Bs, \Cs)$ win if and only if $b_{\Bs} =b_{\Cs} = b$. 
\end{definition}

We denote the advantage of the above game by $\adv_{\As, \Bs, \Cs}(\lambda)$. Note that since $s, s'$ is defined as the canonical vector of both cosets, they are uniquely defined; similarly, $H(s, s')$ is also uniquely defined.

We show the following theorem:
\begin{theorem} \label{thm:moe_rom}
    Let $n = \Omega(\lambda)$, then for every $\lambda \in \mathbb{N}^+$ and all query-bounded algorithms $(\As, \Bs, \Cs)$, $\adv_{\As, \Bs, \Cs}(\lambda) \leq \frac{1}{2} + \negl(\lambda)$. 
\end{theorem}

\subsection{Proof for \texorpdfstring{\Cref{thm:moe_rom}}{Strengthened MOE Game}}
\label{sec:main_proof}

\begin{proof}
    We prove the theorem by following hybrid arguments.
    
    \begin{hybrid}{0}
    \label{hyb:ue0}
    This hybrid is the original game.
    \end{hybrid}
    \begin{hybrid}{1}
    \label{hyb:ue1}
     This hybrid follows \Cref{hyb:ue0}, but the oracle of $\As$ will be reprogrammed as $H_{s,s'}$ defined as follows:
    
    \begin{align*}
                H_{s,s'}(z, z') = \begin{cases}
                        u   & \text{ if } z = s, z' = s' \\
                        H(z, z') & \text{ otherwise }
                    \end{cases},
    \end{align*}
    where $u \in \bit^n$ is chosen uniformly at random.
    \end{hybrid}
    
    \begin{hybrid}{2}
    \label{hyb:ue2}
    This hybrid will modify the access to random oracle of $\Bs$ and $\Cs$.
    \begin{itemize}
    \item The adversary $\As$ generates $\Delta\in\{0,1\}^{n(\lambda)}$ and sends $\Delta$ to the challenger.
    \item The challenger samples a random subspace $A \subseteq \F_2^\lambda$ of dimension $\lambda/2$ and two random vectors $(s,s') \in \canonicalset(A) \times \canonicalset(A^\perp)$. The challenger uniform randomly samples a bit $b\in\{0,1\}$ and $r\in\{0,1\}^{n(\lambda)}$, and defines the oracle $H_{s,s'}^b$ as follows:

    \begin{align*}
                H_{s,s'}^b(z, z') = \begin{cases}
                        r\oplus (b\cdot\Delta)   & \text{ if } z = s, z' = s' \\
                        H(z, z') & \text{ otherwise }
                    \end{cases},
            \end{align*}
            
    It gives $\ket{A_{s, s'}}$ and $r$ to $\As$.
    
    \item $\As, \Bs, \Cs$ get (quantum) oracle access to $P_{A+s}$ and $P_{A^\perp + s'}$. %
    
    \item With access to quantum random oracle $H_{s,s'}$, $\As$ produces a quantum state over registers $\B \C$ and sends $\B$ to $\Bs$ and $\C$ to $\Cs$.
    
    \item With access to quantum random oracle $H_{s,s'}^b$, $\Bs, \Cs$ are given the description of $A$, they try to produce bits $b_{\Bs}, b_{\Cs}$. 
\end{itemize}
$(\As, \Bs, \Cs)$ win if and only if $b_{\Bs} =b_{\Cs} = b$.
    \end{hybrid} 

    We denote by $p_i$ the optimal success probability of the game in \textbf{Hybrid $\mathbf{i}$}. For the relations between different $p_i$, we have following lemmas:
    
    \begin{lemma}
    \label{lem:uehyb01}
        $\abs{p_0-p_1}\le\negl(\lambda)$.
    \end{lemma}
    
    \begin{lemma}
    \label{lem:uehyb12}
        $p_1=p_2$.
    \end{lemma}
    
    \begin{lemma}
    \label{lem:rprom}
        $p_2\leq\frac{1}{2}+\negl(\lambda)$.
    \end{lemma}
    Combining the three lemmas, we have completed the proof of \Cref{thm:moe_rom}.

\end{proof}

Now we provide proofs for lemmas beyond.

\begin{proof}[Proof for \Cref{lem:uehyb01}] 
    We prove by contradiction. Suppose $p_0\geq p_1+1/q(\lambda)$ for some polynomial $q(\lambda)$, then we can construct an adversary $\As'$ that violates the direct product hardness of coset states. %
    $\As'$ will perform as follows:
    \begin{itemize}
        \item $\As'$ 
        samples a random oracle $H:\F_2^\lambda \times \F_2^\lambda \to \{0,1\}^{n(\lambda)}$.
        \item $\As'$ simulates $\As$ using $H$ and applies computational basis measurement on a random quantum query made by $\As$ to the random oracle. 
    \end{itemize} 
    By \Cref{thm:bbbv}, assuming $\As$ makes at most $T$ queries, then $\As'$ gets $(s,s')$ with probability at least $4/(q^2T)$, a contradiction to \Cref{col:dp_negl}.
\end{proof}

\begin{proof}[Proof of \Cref{lem:uehyb12}]
    Fixing $\Delta$ and $b$, the two games are identical by renaming the $w=H(s,s')\oplus(b\cdot\Delta)$ to $r$. Since $H(s,s')$ is uniformly random, its distribution is identical to $r$.
\end{proof}

\begin{proof}[Proof of \Cref{lem:rprom}]

    Fixing $A, r, \Delta$, two canonical vectors $s, s'$, let $H_{-s, s'}$ be a partial random oracle that is defined on every input except $(s, s')$. Fix any partial random oracle $H_{-s, s'}$, 
    we define two \emph{projectors} $\Pi^B_0, \Pi^B_1$ over register $\B$ as:
    \begin{itemize}
        \item $\Pi^B_0$: runs $\Bs$ on input $A$ with oracle access to $ H_{s,s'}^0$ where $H_{s,s'}^0$ is the same as $H_{-s,s'}$ except on input $(s, s')$ it outputs $r$;  it measures if the outcome is $r$; then it undoes all the computation. 
        \item $\Pi^B_1$: similar to $\Pi^B_0$ except on input $(s, s')$, the random oracle $H^1_{s,s'}$ outputs $r\oplus\Delta$ and it checks if the outcome is $r\oplus\Delta$. 
    \end{itemize}
    Let $\{\ket {\phi_i}\}_{i}$ be a set of the eigenvectors of $(\Pi^B_0 + \Pi^B_1)/2$ with eigenvalues $\{\lambda_i\}_i$.

    Fixing the same $A, s, s', r $ and $H_{-s,s'}$, we can similarly define $\Pi^C_0, \Pi^C_1$ for $\Cs$. Let $\{\ket {\psi_j}\}_{j}$ be a set of the eigenvectors of $(\Pi^C_0 + \Pi^C_1)/2$ with eigenvalues $\{\mu_j\}_{j}$. 
    
    \medskip
    
    Let $\ket {\phi_{\B \C}}$ be the state prepared by $\As$. Without loss of generality, we can assume the state is pure. 
    We write the state under the basis $\{\ket {\phi_i}\}_{i}$ and  $\{\ket {\psi_j}\}_{j}$: 
    \begin{align*}
        \ket {\phi_{\B \C}} = \sum_{i, j} \alpha_{i, j} \ket{\phi_i}_\B \otimes \ket{\psi_j}_\C. 
    \end{align*}
    \begin{lemma}\label{lem:smallweight}
        Taken the randomness of $A, s, s'$ and $H_{-s,s'}$, for every polynomial $p(\cdot)$, there exists a negligible function $\negl$ such that with overwhelming probability 
        the following weight is bounded: 
        \begin{align*}
            \sum_{\substack{i: \; |\lambda_i - 1/2| > 1/p \\ j: \; |\mu_j - 1/2| > 1/p}} |\alpha_{i,j}|^2 \leq \negl(n). 
        \end{align*}
    \end{lemma}
    The proof for this lemma is given at the end of this section.

    With the above lemma, we can claim that over the randomness of $A, s, s'$ and $H_{-s, s'}$, for every polynomial $p(\cdot)$, $\ket {\phi_{\B \C}}$ is negligibly close to the following state $\ket {\phi'_{\B \C}}$: 
    \begin{align*}
        \sum_{\substack{i: |\lambda_i - 1/2| \leq 1/p}} \alpha_{i, j} \ket{\phi_i}_\B \otimes \ket{\psi_j}_\C + \sum_{\substack{i: |\lambda_i - 1/2| > 1/p \\ j: |\mu_j - 1/2| \leq 1/p}} \alpha_{i, j} \ket{\phi_i}_\B \otimes \ket{\psi_j}_\C. 
    \end{align*}
    For convenience, we name the left part as $\ket{\phi'_\Bs}$ (indicating $\Bs$ can not win) and the right part as $\ket{\phi'_\Cs}$ (indicating $\Cs$ can not win). Thus, for every polynomial $p(\cdot)$, there exists a negligible function $\negl(\cdot)$, $|\ket{\phi_{\B \C}} - (\ket{\phi'_\Bs} + \ket{\phi'_\Cs})|_1$ is at most $\negl(\cdot)$ (in expectation, taken the randomness of $A, s, s',r$ and $H_{-s, s'}$).
    
    \medskip
    The probability that $(\As, \Bs, \Cs)$ wins is at most:
    \begin{align*}
        (\left| (\Pi^B_0 \otimes \Pi^C_0) \ket {\phi'_{\B \C}}\right|^2 + \left| (\Pi^B_1 \otimes \Pi^C_1) \ket {\phi'_{\B \C}}\right|^2)/2. 
    \end{align*}
    $\Pi^B_0 \otimes \Pi^C_0$ is the case that they both get access to $H_0$ and $\Pi^B_1 \otimes \Pi^C_1$ for $H_1$. 
    
    The probability is at most
    \begin{align*}
        &(\left| (\Pi^B_0 \otimes \Pi^C_0) (\ket{\phi'_\Bs} + \ket {\phi'_\Cs})\right|^2 + \left| (\Pi^B_1 \otimes \Pi^C_1) (\ket{\phi'_\Bs} + \ket {\phi'_\Cs)}\right|^2)/2 \\
        =& \frac{1}{2} \cdot \left( \langle \phi'_{\Bs} | (\Pi^B_0 \otimes \Pi^C_0) | \phi'_{\Bs}  \rangle +  \langle \phi'_{\Bs} | (\Pi^B_1 \otimes \Pi^C_1) | \phi'_{\Bs}  \rangle  + \langle \phi'_{\Cs} | (\Pi^B_0 \otimes \Pi^C_0) | \phi'_{\Cs}  \rangle  +  \langle \phi'_{\Cs} | (\Pi^B_1 \otimes \Pi^C_1) | \phi'_{\Cs}  \rangle \right) \\
        +& \mathsf{Re}\left(  \langle \phi'_{\Bs} | (\Pi^B_0 \otimes \Pi^C_0) | \phi'_{\Cs} \rangle + \langle \phi'_{\Bs} | (\Pi^B_1 \otimes \Pi^C_1) | \phi'_{\Cs}  \rangle \right) \\
        \leq & \frac{1}{2} \cdot \left( \langle \phi'_{\Bs} | (\Pi^B_0 \otimes I) | \phi'_{\Bs}  \rangle +  \langle \phi'_{\Bs} | (\Pi^B_1 \otimes I) | \phi'_{\Bs}  \rangle  + \langle \phi'_{\Cs} | (I \otimes \Pi^C_0) | \phi'_{\Cs}  \rangle  +  \langle \phi'_{\Cs} | (I \otimes \Pi^C_1) | \phi'_{\Cs}  \rangle \right) \\
        +& \mathsf{Re}\left(  \langle \phi'_{\Bs} | (\Pi^B_0 \otimes \Pi^C_0) | \phi'_{\Cs} \rangle + \langle \phi'_{\Bs} | (\Pi^B_1 \otimes \Pi^C_1) | \phi'_{\Cs}  \rangle \right).
    \end{align*}
    We bound each term separately. 
    \begin{itemize}
        \item $\frac{1}{2} \left( \langle \phi'_{\Bs} | (\Pi^B_0 \otimes I) | \phi'_{\Bs}  \rangle +  \langle \phi'_{\Bs} | (\Pi^B_1 \otimes I) | \phi'_{\Bs}  \rangle \right)$. It is equal to $\langle \phi'_{\Bs} | (\Pi^B_0 + \Pi^B_1)/2 \otimes I | \phi'_{\Bs} \rangle$; by the definition of $\ket{\phi'_\Bs}$, it will be at most $(\frac{1}{2} + \frac{1}{p}) |\ket{\phi'_\Bs}|^2$. 
        
        \item $\frac{1}{2} \left( \langle \phi'_{\Cs} | (I \otimes \Pi^C_0) | \phi'_{\Cs}  \rangle  +  \langle \phi'_{\Cs} | (I \otimes \Pi^C_1) | \phi'_{\Cs}  \rangle \right)$. Similar to the above case, it is at most $(\frac{1}{2} + \frac{1}{p}) |\ket{\phi'_\Cs}|^2$.
        
        \item $\mathsf{Re}\left(  \langle \phi'_{\Bs} | (\Pi^B_0 \otimes \Pi^C_0) | \phi'_{\Cs} \rangle \right)$.  By \Cref{cor:orthogonal_eigen}, the inner product will be 0: 
        \begin{align*}
            \langle \phi'_{\Bs} | (\Pi^B_0 \otimes \Pi^C_0) | \phi'_{\Cs} \rangle &= \sum_{i: |\lambda_i - 1/2| \leq 1/p}  \sum_{\substack{i': |\lambda_{i'} - 1/2| > 1/p \\ j': |\mu_{j'} - 1/2| \leq 1/p}}  \alpha^\dagger_{i, j} \alpha_{i', j'} \langle \phi_i | \Pi^B_0 | \phi_{i'} \rangle \langle \psi_j | \Pi^C_0 | \psi_{j'} \rangle; 
        \end{align*}
        since every possible $i, i'$ satisfy $\lambda_i + \lambda_{i'} \ne 1$, we have $\langle \phi_i | \Pi^B_0 | \phi_{i'} \rangle = 0$.
        
        \item $\mathsf{Re}\left(  \langle \phi'_{\Bs} | (\Pi^B_1 \otimes \Pi^C_1) | \phi'_{\Cs} \rangle \right)$.  By \Cref{cor:orthogonal_eigen}, the inner product will be 0 as well. 
    \end{itemize}
    
    Therefore, the total probability will be at most $\left( \frac{1}{2} + \frac{1}{p} \right) (|\ket{\phi'_\Bs}|^2 + |\ket{\phi'_\Cs}|^2) + \negl(n) \leq \frac{1}{2} + \frac{1}{p} + \negl(n)$.
    
    \medskip
    
    Since the above statement holds for every polynomial $p(\cdot)$, it finishes the proof for \Cref{thm:moe_rom}. 
    
\end{proof}

Finally, we give the proof for \Cref{lem:smallweight}. 
    \begin{proof}[Proof of \Cref{lem:smallweight}]
        We prove by contradiction: suppose there exists an adversary $(\As, \Bs, \Cs)$ such that the weight, which we call $W$, is non-negligible, i.e. $W > 1/q(\lambda)$ for some polynomial $q(\cdot)$, with some non-negligible probability $\eta(\secparam)$. For convenience, we will omit $\lambda$ in the proof when it is clear from the context. 
        
        We construct the following adversary $(\As', \Bs', \Cs')$ that breaks the regular MOE game in~\Cref{def:regMOE}:
        \begin{enumerate}
            \item $\As', \Bs', \Cs'$ get (quantum) oracle access to $P_{A+s}$ and $P_{A^\perp + s'}$.

            \item $\As'$ first receives $\Delta$ from simulated $\As$, it samples $r\in\{0,1\}^{n(\lambda)}$ and a random oracle $H$. Given $\ket{A_{s, s'}},r $ and two membership checking oracles, it simulates $\As$ via reprogrammed $H_{s,s'}$, and produces $\ket{\phi_{\B \C}}$; it gives $\B$ to $\Bs'$ and $\C$ to $\Cs'$.

            Note that, although $H$ is a total random oracle, we will later reprogram $H$ at the input $(s, s')$.
            Thus, $H$ will only serve as $H_{-s,s'}$. 
            Since $\As'$ does not know $(s, s')$, it is hard for $\As'$ to only sample $H_{-s,s'}$.

            \item Define two projectors $\Pi^B_0, \Pi^B_1$ over register $\B$ as what we have described at the beginning of the proof, with the random oracle $H_{s,s'}^0$ and $H_{s,s'}^1$ is defined as: 
            \begin{align*}
                H_{s,s'}^0(z, z') = \begin{cases}
                        r   & \text{ if } z = s, z' = s' \\
                        H(z, z') & \text{ otherwise }
                    \end{cases},
            \end{align*}
            and
            \begin{align*}
                H_{s,s'}^1(z, z') = \begin{cases}
                        r\oplus\Delta   & \text{ if } z = s, z' = s' \\
                        H(z, z') & \text{ otherwise }
                    \end{cases}.
            \end{align*}
            
            Given $P_{A+s}, P_{A^\perp+s'}$ and the description of $A$, one can efficiently implement point functions that check the canonical vectors $s$ and $s'$; thus, additionally given $H$, $H_{s,s'}^0$ and $H_{s,s'}^1$ can also be efficiently simulated. Therefore, $\Bs'$ can implement both $\Pi^B_0, \Pi^B_1$ efficiently. 
            
            $\Bs'$ gets $\B$, it applies the efficient approximate threshold measurement $\sati_{(P,Q), \gamma}^{\epsilon, \delta}$ in \Cref{thm:thres_approx}
            with $P = (\Pi^B_0 + \Pi^B_1)/2$, $Q = I - P$, $\gamma = 3/4p$, $\epsilon = 1/4p$ and $\delta = 2^{-\lambda}$.

            If the outcome is 1, $\Bs'$ then runs $\Bs$ on the leftover state with $H_0$ or $H_1$ picked uniformly at random. It measures and outputs a random query $\Bs$ makes to the random oracle. 

            \item Similarly define $\Pi^C_0, \Pi^C_1$ as above on register $\C$. $\Cs'$ gets $\C$, it applies the efficient approximated threshold measurement $\sati_{(P,Q), \gamma}^{\epsilon, \delta}$ with $P = (\Pi^C_0 + \Pi^C_1)/2$, $Q = I - P$,  $\gamma = 3/4p$, $\epsilon = 1/2p$, and $\delta = 2^{-\lambda}$. 
            
            When the outcome is $1$, $\Cs'$ runs $\Cs$ on the leftover state with $H_0$ or $H_1$ picked uniformly at random. It measures and outputs a random query to the random oracle.
        \end{enumerate}

        By \Cref{thm:thres_approx} bullet (1), conditioned on $W \ge 1/q$, both $\Bs'$ and $\Cs'$ will get outcome 1 with probability $1/q - 2\delta = O(1/q)$. When both outcomes are 1, by bullet (2) of \Cref{thm:thres_approx}, the leftover state is $4\delta$-close to the the following state:
        \begin{align*}
            \sum_{\substack{i: |\lambda_i - 1/2| > 1/4p \\ j: |\mu_j - 1/2| > 1/4p}} \beta_{i,j} \ket{\phi_i}_\B \otimes \ket{\psi_j}_\C. 
        \end{align*}
        Observe that when $\Bs$ does not query $(s,s')$, it will succeed with probability exactly $1/2$. Therefore, by \Cref{thm:bbbv}, the query weight of $\Bs$ on $(s,s')$ is at least $1/4p^2T - \negl(\secparam)$, where $T$ is an upper-bound on the number of queries made by $\Bs$.
        Arguing similarly for $\Cs$, we conclude that the adversary $(\As', \Bs', \Cs')$ wins with probability at least $O(\eta/(q p^4 T^2))$, which is non-negligible. 

    \fatih{Updated this proof to align with \Cref{thm:thres_approx}.}
    \end{proof}

\section{Unclonable Encryption in the QROM}
\label{sec:ue_scheme}

The following is the unclonable encryption scheme for a single bit:
\begin{enumerate}
    \item $\sk = A$ where $A$ is a random subspace $A \subseteq \F_2^n$ of dimension $n/2$; 
    \item $\enc^H(\sk, m)$: it samples $s \leftarrow \canonicalset(A)$ and $s' \leftarrow \canonicalset(A^\perp)$ uniformly at random; it outputs $\ket{A_{s, s'}}$, $c = H(s, s') \oplus m$;
    \item $\dec^H(\sk = A, (\ket{A_{s, s'}}, c))$:  
    \begin{itemize}
        \item It first computes $s$ in superposition. We know that there is a classical algorithm that on any vector in $A + s$ and the description of $A$, outputs the canonical vector of $A+s$ (which is $s$ in this case).  See~\cite{coladangelo2021hidden} Definition 4.3 for more references. 
        
        We can run this classical algorithm coherently on $\ket{A_{s, s'}}$ to learn $s$.
        
        \item Since the algorithm on any vector in $A+s$ outputs the same vector, the quantum state stays intact. We can run the same algorithms coherently on the Hadamard basis and the description of $A^\perp$ to learn $s'$.
        
        \item Output $c \oplus H(s, s')$. 
    \end{itemize}
\end{enumerate}

With \Cref{thm:moe_rom}, we can show the scheme satisfy the unclonable IND-CPA security. 

\begin{proof}
If we have some adversary $(\As,\Bs,\Cs)$ for the scheme beyond, we can construct an adversary $(\As',\Bs',\Cs')$ for the strengthened MOE game with the same advantage. 
\begin{itemize}
    \item The adversary $\As'$ gets $(m_0,m_1)\leftarrow \As$ and sends $\Delta=m_0\oplus m_1$ to the challenger.
    \item After receiving $\ket{A_{s,s'}}$ and $w$ from the challenger, $\As'$ calculates $c=w\oplus m_0$, and sends $(\ket{A_{s,s'}},c)$ to $\As$. The output registers $\B, \C$ of $\As$ are sent to $\Bs',\Cs'$ respectively.
    \item $\Bs',\Cs'$ exactly run the algorithm of $\Bs,\Cs$, and output their output respectively.
\end{itemize}

Thus we have concluded the unclonable IND-CPA security of our game.

\end{proof}

\begin{remark}
Notice that compared to the strengthened MOE game, our construction does not provide additional membership checking oracles. 
\end{remark}

\section{Copy-Protection for Point Functions in QROM}
\label{sec:CP_qrom}

\subsection{Copy-Protection Preliminaries}
\label{sec:prelims_cp}
Below we present the definition of a copy-protection scheme.
\begin{definition} [Copy-Protection Scheme] \label{def:copyprotection}
 Let $\fclass = \fclass(\secparam)$ be a class of efficiently computable functions of the form $f: X \to Y$. A copy protection scheme for $\fclass$ is a pair of QPT algorithms $(\copyprotect, \eval)$ such that: \begin{itemize}
    \item \textbf{Copy Protected State Generation:} $\copyprotect(1^\secparam, d_f)$ takes as input the security parameter $1^\secparam$ and a classical description $d_f$ of a function $f \in \fclass$ (that efficiently computes $f$). It outputs a mixed state $\rho_f \in \density{Z}$, where $Z$ is the output register.
    \item \textbf{Evaluation:} $\eval(1^\secparam, \rho, x)$ takes as input the security parameter $1^\secparam$, a mixed state $\rho\in \density{Z}$, and an input value $x\in X$. It outputs a bipartite state $\rho' \otimes \ket{y}\bra{y} \in \density{Z} \otimes \density{Y}$.
    
\end{itemize}
\end{definition}

\noindent We will sometimes abuse the notation and write $\eval(1^\secparam, \rho, x)$ to denote the classical output $y \in Y$ when the residual state $\rho'$ is not significant.

\begin{definition}[Correctness] \label{def:cpcorrectness}
A copy-protection scheme $(\copyprotect,\eval)$ for $\fclass$ is \textbf{$\delta$-correct} if the following holds: for every $x \in X$, $f \in \fclass$,
$$\Pr \left[ f(x) \leftarrow \eval(1^\secparam, \rho_f,x)\ :\ \rho_f \leftarrow \copyprotect(1^{\secparam},d_f) \right] \geq \delta.$$
If $\delta \ge 1 - \negl(\secparam)$, we simply say that the scheme is \textbf{correct}.
\end{definition}

\begin{remark}
When $\delta$ is negligibly close to 1, the evaluation algorithm $\eval$ can be implemented so that it does not disturb the state $\rho_f$. This ensures that $\rho_f$ can be reused polynomially many times with arbitrary inputs.
\end{remark}

We define security via a piracy experiment.

\begin{definition}[Piracy Experiment]
\label{def:piracyexperiment}
A \textbf{piracy experiment} is a security game defined by a copy-protection scheme $(\cp, \eval)$ for a class of functions $\fclass$ of the form $f:X\to Y$, a distribution $\distr_\fclass$ over $\fclass$, and a class of distributions $\distrclass = \{\distrclass(f)\}_{f \in \fclass}$ over $X \times X$. It is the following game between a challenger and an adversary, which is a triplet of algorithms $\abc$: %
\begin{itemize}
    \item \textbf{Setup Phase:} The challenger samples a function $f\leftarrow \distr_\fclass$ and sends $\rho_f \leftarrow \cp(1^\secparam, d_f)$ to $\alice$.
    \item \textbf{Splitting Phase:} $\alice$ applies a CPTP map to split $\rho_f$ into a bipartite state $\rho_{\B\C}$; it sends the $\B$ register to $\bob$ and the $\C$ register to $\charlie$. No communication is allowed between $\bob$ and $\charlie$ after this phase.
    \item \textbf{Challenge Phase:} The challenger samples $(x_B, x_C) \leftarrow \distrclass(f)$ and sends $x_B, x_C$ to $\bob, \charlie$, respectively. %
    \item \textbf{Output Phase:} $\bob$ and $\charlie$ output $y_B \in Y$ and $y_C \in Y$, respectively, and send to the challenger. The challenger outputs 1 if $y_B = f(x_B)$ and $ y_C = f(x_C)$, indicating that the adversary has succeeded, and 0 otherwise. 
\end{itemize}
The bit output by the challenger is denoted by $\pirateexp{\cp}{\eval}{\distr_\fclass}{\distrclass}(1^\secparam, \abc)$.
\end{definition}

As noted by \cite{CMP20}, the adversary can always succeed in this game with probability negligibly close to \[\trivialprob(\distr_\fclass, \distrclass) := \max_{E \in \bracketsC{B,C}} \E_{\substack{f \from \distr_\fclass \\ (x_B,x_C) \from \distrclass(f)}} \max_{y \in Y} \prob \bracketsS{y \given x_E} \]
by sending $\rho_f$ to $\bob$ and have $\charlie$ guess the most likely output $y$ given input $x_C$ (or vice versa). In other words, $\trivialprob$ is the success probability of optimal guessing strategy for one party $E \in \bracketsC{B,C}$ given only the test input $x_E$.

Bounding the success probability of the adversary is bounded by $\trivialprob$ captures the intuition that $\rho_f$ is no more helpful for simultaneous evaluation than a black-box program that could only be given to one party. 

\begin{definition}[Copy-Protection Security] \label{def:cpsecurity}
Let $(\cp, \eval)$ be a copy-protection scheme for a class $\fclass$ of functions $f:X \to Y$. Let $\distr_\fclass$ be a distribution over $\fclass$ and $\distrclass = \left\{\distrclass(f)\right\}_{f \in \fclass}$ a class of distributions over $X$. Then, $(\cp, \eval)$ is called $\left(\distr_\fclass,\distrclass\right)$-\textbf{secure} if there exists a negligible function $\negl$ such that any QPT adversary $\abc$ satisfies \[\prob \left[ b = 1\ :\ b \leftarrow \pirateexp{\cp}{\eval}{\distr_\fclass}{\distrclass}\left(1^\secparam, \abc\right)  \right] \le \trivialprob(\distr_\fclass, \distrclass) + \negl(\secparam) .\]

\end{definition}

\paragraph{Copy Protection for Point Functions}
A \textit{point function} $f_y: \bit^{m} \to \bit$ is of the form \begin{align*}
    f_{y}(x) = \begin{cases} 1, & x = y \\ 0, & x \ne y  \end{cases}.
\end{align*}

When dealing with point functions, the classical description of $f_y$ will simply be $y$, and accordingly the distribution $\distr_\fclass$ over point functions will be represented by a distribution $\distr = \distr_\secparam$ over $\bit^m$. Since copy protection is trivially impossible for a learnable distribution $\distr$, we are going to restrict our attention to unlearnable distributions. 

\begin{definition}\label{def:unlearnable}
A distribution $\distr_\secparam$ over $\bit^m$, with $m = \poly(\secparam)$, is called \subversion{\textit{unlearnable}}\fullversion{\textbf{unlearnable}} if for any query-bounded adversary $\alice^{f_y(\cdot)}$ with oracle access to $f_y(\cdot)$, we have \begin{align*}
    \prob \bracketsS{y' = y: \substack{y \from \distr_\secparam \\ y' \from \alice^{f_y(\cdot)}(1^\secparam) }} \le \negl(\secparam).
\end{align*}

\end{definition}

\begin{definition}[Copy-Protection Security for Point Functions]\label{def:cpsecurityforpf} Let $m = \poly(\secparam)$ and $\fclass$ be the class of point functions $f_y : \bit^m \to \bit$. Let $\distrclass = \bracketsC{\distrclass(f)}_{f \in \fclass}$ be a class of input distributions over $\bit^m \times \bit^m$. A copy protection scheme $(\cp, \eval)$ for $\fclass$ is called $\distrclass$-\textbf{secure} if there exists a negligible function $\negl$ such that $(\cp, \eval)$ is $(\distr_\secparam, \distrclass)$-\textbf{secure} for all unlearnable distributions $\distr_\secparam$ over $\bit^m$.
\end{definition}

\subsection{Construction}

\noindent In this section, we design copy-protection for a class of point functions. We set $n = 2\secparam$ and $d = \secparam$ throughout the section. Our construction will use two hash functions: (a) $\oracleone:\{0,1\}^{\secparam} \rightarrow \{0,1\}^{n \cdot d}$ and (b) $\oracletwo: \F_2^n \times \F_2^n \rightarrow \{0,1\}^{4n+\secparam}$. In the security proof, we will treat $\oracleone$ and $\oracletwo$ as random oracles. We will use $\F_2^n$ and $\bit^n$ interchangeably.
\par We denote the set of all $d$-dimensional subspaces of  $\F_2^n$ by $\subspace_d$.\fullversion{\ We will need the following lemma for correctness.} %
\fullversion{\begin{lemma}
\label{lem:cp:useful}
There exists a set of efficient unitaries $\{U_{A'}\}_{A' \in \subspace_d} \subseteq \U \brackets{\Hs_{\bf X} \otimes \Hs_{\bf Z} \otimes \Hs_{\bf anc}}$, where ${\bf X}, {\bf Z}, {\bf anc}$ are registers of length $n, 2n, \poly(\secparam)$, such that the following holds for any $A \in \subspace_d$:  %
\begin{itemize}

    \item For any $s \in \canonicalset(A), s' \in \canonicalset(A^\perp)$, we have $U_{A} \ket{A_{s,s'}}\ket{0^{2n}}_{\bf Z}\ket{0^{\poly(\secparam)}}_{{\bf anc}} \allowbreak =\ket{A_{s,s'}}\ket{s,s'}_{\bf Z}\ket{0^{\poly(\secparam)}}_{{\bf anc}}$.
    \item For any $A' \in \subspace_d$ such that $\frac{|A' \cap A|}{2^d} \leq \nu(\secparam)$, for some negligible function $\nu(\cdot)$, there exists a negligible function $\nu'(\secparam)$ such that the following holds for all $s \in \canonicalset(A), s' \in \canonicalset(A^\perp)$:
    $$ \left\|  (I_{\bf X} \otimes \ketbra{s,s'}{s,s'}_{\bf Z} \otimes I_{\bf anc} )\left( U_{A'} \ket{A_{s,s'}}\ket{0^{2n}}_{\bf Z}\ket{0^{\poly(\secparam)}}_{{\bf anc}}\right) \right\|^2 \leq \nu'(\secparam).$$
\end{itemize}
\end{lemma}
\begin{proof}
To get unitaries satisfying the first bullet, recall that there exists an efficient procedure which computes $\can_A(\cdot)$ given the description of $A$. We can represent this procedure by a unitary $U$ followed by measurement of $s,s'$. We describe $U_A$ as follows: \begin{enumerate}
    \item Apply $U$ to the ${\bf X}, {\bf anc}$ registers. Copy the answer to the first half of the ${\bf Z}$ register. Note that the answer is always $s$ since $\ket{A_{s,s'}}$ is a superposition of vectors in $A + s$.
    \item Apply $U^\dagger$ to the ${\bf X}, {\bf anc}$ registers.
    \item Apply QFT on the ${\bf X}$ register to obtain $\ket{A^\perp_{s',s}}$.
    \item Repeat the first two steps and copy the answer $s'$ to the second half of the ${\bf Z}$ register.
    \item Appy QFT again to recover $\ket{A_{s,s'}}_{\bf X}$
\end{enumerate}

\par We will show that the second bullet follows from the first bullet. We first observe that the inner product between the coset states $\ket{A_{s,s'}}$ and $\ket{A'_{s,s'}}$ is small. Indeed, since $\abs{(A + s) \cap (A' + s)} = \abs{A \cap A'} \le 2^d \nu(\secparam)$, we have \begin{align*}
    \abs{\braket{A_{s,s'} | A'_{s,s'}}}^2 &= \abs{\frac{1}{\sqrt{|A|}} \sum_{a \in A} (-1)^{\langle s', a\rangle} \bra{a + s} \frac{1}{\sqrt{|A|}} \sum_{a' \in A'} (-1)^{\langle s', a'\rangle} \ket {a' + s}}^2 \\
    &\le \abs{\frac{1}{|A|} 2^d \nu(\secparam) }^2 = \nu(\secparam)^2.
\end{align*}
Fix $U_{A'}$. Recall that the coset states $\bracketsC{\ket{A'_{t,t'}}}_{t \in \canonicalset(A),t' \in \canonicalset(A^\perp)}$ form an orthonormal basis. By the first bullet, we have \begin{align*}
    & \left\|  (I_{\bf X} \otimes \ketbra{s,s'}{s,s'}_{\bf Z} \otimes I_{\bf anc} ) \left( U_{A'} \ket{A_{t,t'}}_{\bf X} \ket{0^{2n}}_{\bf Z} \ket{0^{\poly(\secparam)}}_{{\bf anc}}  \right)\right\|^2 \\ = \; &\left\| (I_{\bf X} \otimes \ketbra{s,s'}{s,s'}_{\bf Z} \otimes I_{\bf anc} )\left( \ket{A_{t,t'}}_{\bf X} \ket{t,t'}_{{\bf Z}} \ket{0^{\poly(\secparam)}}  \right)\right\|^2 = 0
\end{align*}
for any $(t,t') \ne (s,s')$. Therefore, we have \begin{align*}
    &\left\|  (I_{\bf X} \otimes \ketbra{s,s'}{s,s'}_{\bf Z} \otimes I_{\bf anc} )\left( U_{A'} \ket{A_{s,s'}}_{\bf X} \ket{0^{2n}}_{\bf Z} \ket{0^{\poly(\secparam)}}_{{\bf anc}}  \right)\right\|^2 \\
    &= \bigg\| \sum_{\substack{t \in \canonicalset(A) \\ t' \in \canonicalset(A^\perp)}}   (I_{\bf X} \otimes \ketbra{s,s'}{s,s'}_{\bf Z} \otimes I_{\bf anc} )\left( U_{A'} \ket{A'_{t,t'}}\braket{A'_{t,t'} | A_{s,s'}}_{\bf X} \ket{0^{2n}}_{\bf Z} \ket{0^{\poly(\secparam)}}_{{\bf anc}}  \right)\bigg\|^2
    \\
    &=  \abs{\braket{A_{s,s'} | A'_{s,s'}}}^2 \le \nu(\secparam)^2.
\end{align*}
as desired.
\end{proof}
}

\fullversion{\paragraph{Construction.}}We describe the copy-protection scheme $(\copyprotect,\eval)$ for a class of point functions ${\cal F}=\{f_y(\cdot)\}_{y \in \{0,1\}^{\secparam}}$ as follows: 
\begin{itemize}
    \item $\copyprotect \left( 1^{\secparam},y \right)$: it takes as input $\secparam$ in unary notation, $y \in \{0,1\}^{\secparam}$ and does the following: 
    \begin{enumerate}
        \item \label{item:linearindependence} Compute ${\bf v}=\oracleone(y)$. Parse ${\bf v}$ as a concatenation of $d$ vectors $v_1,\ldots,v_d$, where each $v_i$ has dimension $n$. Abort if the vectors $\{v_1,\ldots,v_d\}$ are not linearly independent. 
        \item Let $A=\mathrm{Span}\left(v_1,\ldots,v_d\right)$. 
           \item Sample $s \from \canonicalset(A)$ and $s' \from \canonicalset(A^\perp)$ uniformly at random. 
        \item \label{item:finalcpstep} Output the copy-protected state $\sigma = \ketbra{A_{s,s'}}{A_{s,s'}}_{\bf X} \otimes \ketbra{\oracletwo(s,s')}{\oracletwo(s,s')}_{\bf Y}$. 
    \end{enumerate}

    \item $\eval(\sigma,x)$: on input the copy-protected state $\sigma \in \mathcal{D}(\Hs_{\bf X} \otimes \Hs_{\bf Y})$, input $x \in \{0,1\}^{\secparam}$, it does the following: 
    \begin{enumerate}
        \item Measure the register ${\bf Y}$ of $\sigma$ to obtain the value $\theta$. Call the resulting state $\sigma'$. 
        \item Compute ${\bf v}=\oracleone(x)$. Parse ${\bf v}$ as a concatenation of $d$ vectors $v_1,\ldots,v_d$, where each $v_i$ has dimension $n$. Abort if the vectors $\{v_1,\ldots,v_d\}$ are not linearly independent.
        \item Let $A=\mathrm{Span}\left(v_1,\ldots,v_d\right)$. 
        \item Apply $U_{A}$ (defined in~\Cref{lem:cp:useful}) coherently on\subversion{\\} $\sigma' \otimes \allowbreak \ketbra{0^{2n}}{0^{2n}}_{\bf Z} \otimes \allowbreak \ketbra{0^{\poly(\secparam)}}{0^{\poly(\secparam)}}_{{\bf anc}}$ to obtain the state $\sigma''$.
        \item Query $\oracletwo$ on the register ${\bf Z}$ and store the answer in a new register ${\bf out}$. 
        \item Measure the register ${\bf out}$ in the computational basis. Denote the post-measurement state by $\sigma_{{\bf out}}$ and the measurement outcome by $\theta'$. 
        \item If $\theta=\theta'$, output $\sigma_{{\bf out}} \otimes \ketbra{1}{1}$. Otherwise, output $\sigma_{{\bf out}} \otimes \ketbra{0}{0}$.
    \end{enumerate}
\end{itemize}

\par We first discuss at a high level why this construction works. Regarding correctness, we argue that $\eval$ on input $x \ne y$ computes a random subspace $A'$, such that $\ket{A'_{s,s'}}$ is nearly orthogonal to $\ket{A_{s,s'}}$. As a result, $\eval$ recovers $(s,s')$ incorrectly. Since as a sufficiently expanding hash function $H$ is injective with high probability, $\eval$ fails. \\
\par As for security, first we show that it is hard for $\alice$ to query the oracles $\oracleone, \oracletwo$ on inputs $y, (s,s')$. Next, we argue that $\bob$ and $\charlie$ cannot both recover $(s,s')$, otherwise they break the MOE game in \Cref{thm: monogamy info}. \\

\par Most meaningful input distributions $\distrclass(y)$ for a point function $f_y$ can be parameterized by a triple $(p,q,r)$: \begin{itemize}
    \item With probability $p$, output $(y,y)$
    \item With probability $q$, output $(y, x_C)$, where $x_C \ne y$ is a random string.
    \item With probability $r$, output $(x_B, y)$, where $x_B \ne y$ is a random string.
    \item With probability $1 - p - q - r$, output $(x_B, x_C)$, where $x_B, x_C \ne y$ are random strings.
\end{itemize}

We show that our scheme is secure with respect to product distributions, i.e. when $(p,q,r,1-p-q-r)$ is of the form $(pp',pq',qp',qq')$ with $p + q = p' + q' = 1$, in \Cref{lem:cpsecurity}. We also show security for maximally correlated input distributions, i.e. when $q = r = 0$, in \Cref{cor:cp_correlated}. The way the random strings $x_B, x_C$ are sampled (uniformly or otherwise) turns out to be inconsequential in our security proof. \fatih{Added this paragraph to clarify the result statements.} \\

We give the formal statements below\subversion{ proofs in \Cref{sec:proofs}}. 

\begin{lemma} \label{lem:cp_correctness}
$(\copyprotect,\eval)$ satisfies correctness. 
\end{lemma}
\fullversion{\subversion{ We first prove a useful lemma.

}
\begin{proof}[Proof of \Cref{lem:cp_correctness}]
\newcommand{\TD}{\mathsf{TD}}
\noindent We first argue that step \ref{item:linearindependence} of $\copyprotect$ aborts only with negligible probability:

\begin{claim}
\label{clm:linearindependence}
Let $n=2d=2\secparam$ and $v_1, v_2, \dots, v_d \in \F_2^n$ be uniformly random independent vectors, then there exists a negligible function $\nu_0$ such that $v_1, \dots, v_d$ are linearly dependent with probability at most $\nu_0(\secparam)$.
\end{claim}

\begin{proof}
Let $p_i$ be the probability that $\bracketsC{v_1, \dots, v_{i}}$ is linearly independent given that $\bracketsC{v_1, \dots, v_{i-1}}$ is linearly independent. Since $v_i$ is uniformly random and the span of $\bracketsC{v_1, \dots, v_{i-1}}$ has size $2^{i-1}$, we have $p_i = 1 - 2^{i-1}/2^n$. Thus, the probability that $\bracketsC{v_1, \dots, v_d}$ is linearly independent is given by \begin{align*}
    \prod_{i=1}^d p_i = \prod_{i=1}^d \brackets{1 - 2^{i-1-n}} \ge \brackets{1 - 2^{-\secparam}}^\secparam \ge 1 - \secparam2^{-\secparam},
\end{align*}
where we used the union bound in the last step. Hence, the claim holds for $\nu_0(\secparam) = \secparam 2^{-\secparam}$.
\end{proof}

\noindent We will condition on step \ref{item:linearindependence} of $\cp$ not aborting henceforth.
Let $y \in \bit^{\secparam}$ and $\sigma \leftarrow \copyprotect\left( 1^{\secparam},y \right)$. Note that $\sigma$ is of the form $\ketbra{A_{s,s'}}{A_{s,s'}}_{\bf X} \otimes \ketbra{\theta}{\theta}_{\bf Y}$, where the following holds: 

\begin{enumerate}
     \item ${\bf v}=\oracleone(y)$ and ${\bf v}$ is a concatenation of $d$ linearly independent vectors $v_1,\ldots,v_d$
        \item $A=\mathrm{Span}\left(v_1,\ldots,v_d\right)$
           \item $s \in \canonicalset(A)$ and $s' \in \canonicalset(A^\perp)$ are selected uniformly at random.
           \item \label{item:secondhashstep} $\theta=\oracletwo(s,s')$
\end{enumerate}

 We now consider the two cases: $x = y$ and $x \ne y$. \\

\noindent {\em {\bf Case 1}. $\eval(\sigma,y) = \sigma_{{\bf out}} \otimes \ketbra{1}{1}$, for some state $\sigma_{{\bf out}}$.} If we follow the first four steps of $\eval(\sigma,y)$, we will end up with the subspace $A$ (defined above). From~\Cref{lem:cp:useful}, we have the following: after applying $U_A$ on $\ket{A_{s,s'}}\ket{0}_{{\bf anc}}$, we obtain $(s,s')$ in ${\bf anc}$ register. That is, ${\bf anc}$ register has the state $\ket{s,s'}$. After querying $\oracletwo$ on ${\bf anc}$, the value stored in ${\bf out}$ is $\oracletwo(s,s')$. Thus, measuring the register ${\bf out}$ yields the value $\theta' = H(s,s')$. Since $\theta'=\theta$, the output of $\eval(\sigma,y)$ is $\sigma_{{\bf out}} \otimes \ketbra{1}{1}$, where $\sigma_{{\bf out}}$ is the residual state. \\

\noindent {\em {\bf Case 2}. $\forall x\neq y$,  $\tracedist{\eval(\sigma,x)}{\sigma_{{\bf out}} \otimes \ketbra{0}{0}} \leq \nu(\secparam)$, for some negligible function $\nu(\secparam)$ and some state $\sigma_{{\bf out}}$.} To prove this, it suffices to show that the probability that $\eval(\sigma,x)$ outputs 1 is negligible in $\secparam$. Consider the following claim:

\begin{claim}
\label{clm:copyprotect:goodH}
If $\oracletwo: \bit^{2n} \to \bit^{4n+\secparam}$ is picked uniformly at random, the probability that $\oracletwo$ is not injective is at most $\nu_1(\secparam)$, for some negligible function $\nu_1(\cdot)$.
\end{claim}
\begin{proof}
For any $a\ne b \in \{0,1\}^{2n}$, the probability that $\oracletwo(a)=\oracletwo(b)$ is $\frac{1}{2^{4n+\secparam}}$. By a union bound argument, the probability that $\oracletwo$ is not injective is at most $\frac{ {2^{2n} \choose 2} }{2^{4n + \secparam}} \leq \frac{2^{4n}}{2^{4n+\secparam}} = \frac{1}{2^{\secparam}}$. 
\end{proof}

\noindent Let us condition on the event that $\oracletwo$ is injective. We consider the first four steps of execution of $\eval(\sigma,x)$:
\begin{itemize}
        \item Measure the register ${\bf Y}$ of $\sigma$ to obtain the value $\theta$. Call the post-measurement state $\sigma'$. 
        \item Compute ${\bf v}=\oracleone(x)$. Parse ${\bf v}$ as a concatenation of $d$ vectors $v_1,\ldots,v_d$, where each $v_i$ has dimension $n$. Abort if the vectors $\{v_1,\ldots,v_d\}$ are not linearly independent.
        \item Let $A'=\mathrm{Span}\left(v_1,\ldots,v_d\right)$.
\end{itemize}
\noindent Consider the following claim. 

\begin{claim}
\label{clm:copyprotect:goodintersect}
If $x \ne y$, then there exists a negligible function $\nu_2(\secparam)$ such that the probability (over the coins of $\oracleone$) that $\frac{|A' \cap A|}{2^d} \leq \nu_2(\secparam)$ holds is at least $1-\nu_2(\secparam)$.
\end{claim}
\begin{proof}
Since $x \ne y$ and $\oracleone$ is a random oracle, $A$ and $A'$ are independently sampled. By \Cref{clm:linearindependence}, $A$ and $A'$ are uniformly random independent subspaces of dimension $d$ each with probability at least $1 - 2\nu_0(\secparam)$. Conditioned on this, we can bound the expected size of their intersection as \begin{align}
    \label{eq:expectationbound}
    \E \bracketsS{|A \cap A'|} = \sum_{v \in \F_2^n} \Pr \bracketsS{v \in A \cap A'} = \sum_{v \in \F_2^n} \Pr \bracketsS{v \in A}^2 = 1 + \brackets{2^n - 1} \brackets{2^{d-n}}^2 < 2.
\end{align}
Let $\nu_2(\secparam) = 2^{-\secparam/5} + 2\nu_0(\secparam)$. Then, by Markov's Inequality and \cref{eq:expectationbound} we have \begin{align*}
    \Pr \bracketsS{\frac{|A' \cap A|}{2^d} > \nu_2(\secparam)} \le 2\nu_0(\secparam) +  \frac{\E \bracketsS{|A' \cap A|} }{2^d\nu_2(\secparam)} < 2\nu_0(\secparam) +  2^{-\secparam/2} < \nu_2(\secparam).
\end{align*}

\end{proof}

\noindent We will condition on the event that $\frac{|A' \cap A|}{2^d} \leq \nu_2(\secparam)$. By ~\Cref{lem:cp:useful}, we have that

 $$p := \left\|  (I_{\bf X} \otimes \ketbra{s,s'}{s,s'}_{\bf Z} \otimes I_{\bf anc} )\left( U_{A'} \ket{A_{s,s'}}\ket{0^{2n}}_{\bf Z}\ket{0^{\poly(\secparam)}}_{{\bf anc}}\right) \right\|^2\leq \nu_3(\secparam). $$
 
 for some negligible function $\nu_3(\secparam)$. Since we have conditioned on the event that $\oracletwo$ is injective, the probability that $\eval(1^\secparam, \sigma, x)$ outputs 1 is given by 
 $$\left\|  (I_{\bf X,Y,Z,anc} \otimes \ketbra{\oracletwo(s,s')}{\oracletwo(s,s')}_{{\bf out}} ) \left(I_{\bf X,Z,anc} \otimes O^{\oracletwo}\right) U_{A'} \ket{A_{s,s'}}_{\bf X} \ket{0^{\poly(\secparam)}}_{{\bf anc}} \ket{0}_{\bf out}  \right\|^2 = p,$$
 where $O^{\oracletwo}$ is the unitary that computes $\oracletwo$. Combining this with \Cref{clm:linearindependence},~\Cref{clm:copyprotect:goodH} and~\Cref{clm:copyprotect:goodintersect}, we conclude that for any $x \ne y$, \begin{align*}
     \prob \bracketsS{ 1 \from \eval(1^\secparam, \sigma, x)} &\le \prob \bracketsS{1 \from \eval(1^\secparam, \sigma, x) \given \substack{ \frac{|A \cap A'|}{2^d} \le \nu_2(\secparam), \; H_2 \text{ is injective }, \\ \cp(1^\secparam, y) \text{ or } \eval(1^\secparam, \sigma, x) \text{ doesn't abort} }} \\
     &+ \prob \bracketsS{ \cp(1^\secparam, y) \text{ or } \eval(1^\secparam, \sigma, x) \text{ aborts} } + \prob \bracketsS{H_2 \text{ is not injective} } \\
     &+ \prob \bracketsS{ \frac{|A \cap A'|}{2^d} > \nu_2(\secparam) \given H_2 \text{ is injective}} \\
     &\le \nu_3(\secparam) + 2\nu_0(\secparam) + \nu_1(\secparam) + \nu_2(\secparam) \\
     &\le \negl(\secparam).
 \end{align*}

\end{proof}
}

\begin{lemma} \label{lem:cpsecurity}
$(\copyprotect,\eval)$ is a $\distrclass$-secure copy-protection scheme for point functions with input length $\secparam$, where $\distrclass(y) = \distrmarg_y^B \times \distrmarg_y^C$ is a product distribution.
\end{lemma}

\fullversion{%
\begin{proof}[Proof of \Cref{lem:cpsecurity}]

Fix an unlearnable distribution $\distr = \distr_\secparam$. We will define a sequence of hybrids:

\begin{hybrid}{1} \label{hyb:cp1}
This is the real piracy experiment for $(\cp, \eval)$ defined in \Cref{def:piracyexperiment}, where $\abc$ all have access to both random oracles $\oracleone$ and $\oracletwo$. The input to $\cp$ is denoted by $y \in \bit^\secparam$ as in the construction.
\end{hybrid}

\begin{hybrid}{2} \label{hyb:cp2}
In this hybrid, we change, for $\alice$ only, the oracle $\oracleone$ to $\oracleone_y$, which is the punctured oracle defined as
\begin{align*}
    \oracleone_y(x) = \begin{cases} u, & x = y \\ \oracleone(x), & x \ne y \end{cases},
\end{align*}
where $u \in \bit^{nd}$ is a fresh uniformly random string.

\end{hybrid}

\begin{hybrid}{3} \label{hyb:cp3}
In this hybrid, we have the challenger sample $A \subseteq \subspace_d$ uniformly at the start. Using this $A$, we change the oracle $\oracleone$ for $\bob$ and $\charlie$ both to $\repryA$, which is the reprogrammed oracle defined as follows: \begin{itemize}
    \item Fix a random basis $\brackets{v_1, \dots, v_d}$ of $A$.
    \item If $x = y$, then $\repryA(x)$ outputs $\brackets{v_1, \dots, v_d}$.
    \item If $x \ne y$, then $\repryA(x)$ outputs $\oracleone(x)$.
\end{itemize}

\end{hybrid}

\begin{hybrid}{4} \label{hyb:cp4}
In this hybrid, we change, for $\alice$ only, the oracle $\oracletwo$ to the punctured oracle $\oracletwo_{s,s'}$ defined as \begin{align*}
    \oracletwo_{s,s'}(t,t') = \begin{cases} v, & (t,t') = (s,s') \\ \oracletwo(t,t'), & (t,t') \ne (s,s') \end{cases},
\end{align*}
where $v \in \bit^{4n+\secparam}$ is a fresh uniformly random string.
\end{hybrid}

Let $p_i$ be the probability that $(\As, \Bs, \Cs)$ wins in Hybrid $i$, and let $\trivialprob = \trivialprob(\distr_\secparam, \distrclass)$. We will show the following lemmas about the hybrids:

\begin{lemma} \label{lem:cphyb12}
$\abs{p_1 - p_2} \le \negl(\secparam)$.
\end{lemma}

\begin{lemma} \label{lem:cphyb23}
$|p_2 - p_3| \le \negl(\secparam)$.
\end{lemma}

\begin{lemma} \label{lem:cphyb34}
$\abs{p_3 - p_4} \le \negl(\secparam)$.
\end{lemma}

\begin{lemma} \label{lem:cphyb4}
$p_4 \le \trivialprob + \negl(\secparam).$
\end{lemma}

\begin{proof}[Proof of \Cref{lem:cphyb12}]

Let $\rho_{\B\C}^{(i)}$ be  the bipartite state sent by $\alice$ to $\bob$ and $\charlie$ in the $i$th Hybrid. We will show that $\tracedist{\rho_{\B\C}^{(1)}}{\rho_{\B\C}^{(2)}} \le \negl(\secparam)$. 
Since \Cref{hyb:cp1} and \Cref{hyb:cp2} are identical after the splitting phase and trace distance cannot increase by post-processing, this suffices to prove the lemma.

\par Suppose that $\tracedist{\rho_{\B\C}^{(1)}}{\rho_{\B\C}^{(2)}}$ is non-negligible. Using $\alice$ from \Cref{hyb:cp2} we will construct an adversary $\aliceprime$ which violates the unlearnability of $\distr_\secparam$ (\Cref{def:unlearnable}) without using the oracle $f_y(\cdot)$. \begin{itemize}
    \item $\aliceprime$ samples random oracles $\oracleone, \oracletwo$, a random subspace $A \in S_d$, and random $(s,s') \in \canonicalset(A) \times \canonicalset(A^\perp)$. 
    \item $\aliceprime$ runs $\alice$ on input $\brackets{\oracleone, \oracletwo, \ket{A_{s,s'}}, H(s,s')}$. Then it measures a random query $y'$ made by $\alice$ to $\oracleone$, and outputs $y'$. 
\end{itemize}
By \Cref{thm:bbbv}, the probability $\Pr\bracketsS{y' = y}$ is non-negligible, thus $\aliceprime$ breaks unlearnability.

\end{proof}

\begin{proof}[Proof of \Cref{lem:cphyb23}]
This easily follows by the fact that \Cref{hyb:cp2} and \Cref{hyb:cp3} are identical conditioned on the fact that $\oracleone(y)$ outputs a valid basis, which happens with overwhelming probability by \Cref{clm:linearindependence}.
\end{proof}

\begin{proof}[Proof of \Cref{lem:cphyb34}]
Similarly as before, it suffices to show $\tracedist{\rho_{\B\C}^{(1)}}{\rho_{\B\C}^{(2)}} \le \negl(\secparam)$. Suppose this is not the case, we will construct an adversary $\aliceprime$ which breaks direct product hardness (\Cref{col:dp_negl}) using $\alice$ from \Cref{hyb:cp4}: \begin{itemize}
    \item $\aliceprime$ receives $\ket{A_{s,s'}}$ from the challenger, where $(s,s') \in \canonicalset(A) \times \canonicalset(A^\perp)$. It samples random oracles $\oracleone,\oracletwo$, and a random string $v \in \bit^{4n+\secparam}$.
    \item $\aliceprime$ runs $\alice$ on input $\brackets{\oracleone, \oracletwo, \ket{A_{s,s'}}, v}$. It measures and outputs a random query $(t,t')$ made to $\oracletwo$ during the execution.
\end{itemize} 
By \Cref{thm:bbbv}, the probability $\Pr\bracketsS{(t,t') = (s,s')}$ is non-negligible, thus $\aliceprime$ breaks direct product hardness.

\end{proof}

\begin{proof}[Proof of \Cref{lem:cphyb4}]
We will use the same template as in the proof of \Cref{lem:rprom}. Let $\rho_{\B\C} := \rho_{\B\C}^{(4)}$ be the bipartite state created by $\alice$. We can assume without loss of generality that $\rho_{\B\C} := \ketbra{\phi_{\B\C}}{\phi_{\B\C}}$ is a pure state. Define POVM elements $\Pi^B, \Pi^C$ as follows:
\begin{itemize}
    \item $\Pi^B$: samples $x_B \from \distrmarg_y^B$; it runs $\bob$ on input oracles $\oracleone_y^A, \oracletwo$ and test input $x_B$; it measures if the output is $f_y(x_B)$; then it undoes all the computation.
    \item $\Pi^C$: defined similarly for $\charlie$.
\end{itemize}

\noindent Now we write the state in its spectral decomposition \begin{align*}
    \ket{\phi_{\B\C}} = \sum_{i,j} \alpha_{i,j} \ket{\phi_i}_\B\ket{\psi_j}_\C,
\end{align*}
where $\ket{\phi_i}_\B$ is an eigenvector of $\Pi^B$ with eigenvalue $\lambda_i$ and $\ket{\psi_j}_\C$ is an eigenvector of $\Pi^B$ with eigenvalue $\mu_j$. Let $\trivialprob_B$ ($\trivialprob_C$) be the trivial guessing probability when $\bob$ ($\charlie$) makes a blind guess, so that $\trivialprob = \max(\trivialprob_B, \trivialprob_C)$. We will need a lemma similar to \Cref{lem:smallweight}.

\begin{lemma} \label{lem:cpsmallweight}
    Let $p(\cdot)$ be a polynomial. With overwhelming probability over $y,(s,s'),\allowbreak \oracleone_y,\allowbreak \allowbreak \oracletwo_{s,s'},\allowbreak \oracletwo(s,s')$, and $A$, we have \begin{align*}
        \sum_{\substack{i: \; \lambda_i > \trivialprob_B + 1/p \\ j: \; \mu_j > \trivialprob_C + 1/p}} \abs{\alpha_{i,j}}^2 \le \negl(\secparam).
    \end{align*}
\end{lemma}

Using this lemma, we can bound the success probability of $\abc$ as \begin{align*}
    &\braket{ \phi_{\B\C} | (\Pi^B \otimes \Pi^C) | \phi_{\B\C}} = \sum_{i,j} \abs{\alpha_{i,j}}^2 \lambda_i \mu_j \\
    &\le \sum_{\substack{i: \; \lambda_i > \trivialprob_B + 1/p \\ j: \; \mu_j > \trivialprob_C + 1/p}} \abs{\alpha_{i,j}}^2 \lambda_i \mu_j \\ &+ \brackets{\trivialprob_B + \frac{1}{p}}\sum_{\substack{i: \; \lambda_i \le \trivialprob_B + 1/p \\ j: \; \mu_j > \trivialprob_C + 1/p}} \abs{\alpha_{i,j}}^2 + \brackets{\trivialprob_C + \frac{1}{p}} \sum_{i, j: \; \mu_j \le \trivialprob_C + 1/p} \abs{\alpha_{i,j}}^2 \\
    &\le \trivialprob + \frac{1}{p} + \negl(\secparam).
\end{align*}
Since $p(\cdot)$ was chosen as an arbitrary polynomial, this suffices for the proof. 

\begin{proof}[Proof of \Cref{lem:cpsmallweight}]%
Suppose for the sake of contradiction that the sum of weights is non-negligible with non-negligible probability over the randomness of $y,(s,s'),\allowbreak \oracleone_y,\allowbreak \oracletwo_{s,s'},\allowbreak \oracletwo(s,s')$. We will construct an adversary $\abcprime$ that breaks the MOE game (\Cref{def:regMOE}): \begin{itemize}
    \item $\abcprime$ get oracle access to $P_{A+s}, P_{A^\perp + s'}$ and $\aliceprime$ receives the state $\ket{A_{s,s'}}$ from the challenger. 
    \item $\aliceprime$ uniformly samples random oracles $\oracleone',\oracletwo'$, as well as random strings $y \from \distr_\secparam, v \in \bit^{4n+\secparam}$. It runs $\alice$ on input $\brackets{\oracleone', \oracletwo', \ket{A_{s,s'}}, v}$ to obtain the bipartite state $\rho_{\B\C}$. It sends the $\B$ register to $\bobprime$ and the $\C$ register to $\charlie$. It also sends $(\oracleone', \oracletwo', y,v)$ to both $\bobprime$ and $\charlieprime$. 
    \item In the second phase, $\bobprime$ and $\charlieprime$ learn the description of $A$. Define the binary POVM elements $\Pi^B, \Pi^C$ over registers $\B, \C$ as above. Note that $\Pi^B$ and $\Pi^C$ are mixtures of projections since one can sample from $\distr_\secparam$ using classical randomness, so that they satisfy the condition of \Cref{thm:thres_approx_asymmetric}. \\
    We observe that $\bobprime$ can efficiently implement $\Pi^B$ as follows: \begin{itemize}
        \item Sample $x_B \from \distrmarg_y^B$ and reprogram $\oracleone'$ on input $y$ to output $A$, obtaining $(\oracleone'_y)^A$.
        \item Sample uniformly $(s,s') \from \canonicalset(A) \times \canonicalset(A^\perp)$ and reprogram $\oracletwo'$ on input $(s,s')$ to output $v$ using the membership oracles $P_{A+s}, P_{A^\perp + s'}$, obtaining $(\oracletwo'_{s,s'})^v$.
        \item Run $\bob$ using the reprogrammed oracles $(\oracleone'_y)^A, (\oracletwo'_{s,s'})^v$ and test input $x_B$.
        \item Measure $\bob$'s output $z$. Undo all the computation.
        \item If $z = f_y(x_B)$, output 1; otherwise output 0.
    \end{itemize} Using this, $\bobprime$ applies the efficient approximated threshold measurement $\ati_{(P,Q),\gamma_1}^{\epsilon, \delta}$ in \Cref{thm:thres_approx_asymmetric}
    with $P = \Pi^B, Q = I - \Pi^B, \gamma_1 = \trivialprob_B + 3/4p, \epsilon = 1/4p$, and $\delta = 2^{-\secparam}$, with outcome $b_B$. \\
    \noindent If $b_B = 0$, $\bobprime$ aborts. If $b_B = 1$, then $\bobprime$ runs a \textit{test execution} on $\bob$, described as follows: $\bobprime$ runs the first three steps of $\Pi^B$ above on $\bob$, and measures a random query $(t_B, t_B')$ made by $\bob$ during the third step to the oracle $(\oracletwo_{s,s'}')^v$. Then, $\bobprime$ outputs $(t_B, t_B')$. We define $\charlieprime$ symmetrically, so that it will measure $b_C$, and if $b_C = 1$ output a query $(t_C, t_C')$ made by $\charlie$ in the test execution.
\end{itemize}

By \Cref{thm:thres_approx_asymmetric} bullet (1), $b_B = b_C = 1$ with non-negligible probability. We will finish the proof by showing that $\bobprime$ and $\charlieprime$ both output $(s,s')$ with non-negligible probability conditioned on $b_B = b_C = 1$. Note that we can intertwine the order of local operations between the two registers this way thanks to no-signalling.

If $b_B = b_C = 1$, then by \Cref{thm:thres_approx_asymmetric} bullet (2) the post-measurement state is negligibly close to a state of the form \begin{align*}
            \sum_{\substack{i:\; \lambda_i > 1/2 + 1/4p \\ j:\; \mu_j > 1/2 + 1/4p}} \beta_{i,j} \ket{\phi_i}_\B \otimes \ket{\psi_j}_\C. 
        \end{align*}

Therefore, in the \textit{test execution}, if $\bobprime$ had not measured $(t_B, t_B')$ in the third step, $\bob$ would correctly output $f_y(x_B)$ correctly with probability greater than $\trivialprob_B + 1/4p$. Consider a modified adversary $\widetilde{\bobprime}$ which is identical to $\bobprime$ except it uses the oracle $\oracletwo'$ (without reprogramming) when running $\bob$. We claim that if $\bobprime$ is replaced by $\widetilde{\bobprime}$, then $\bob$ would output $f_y(x)$ correctly with probability at most $\trivialprob_B$ at the end of the \textit{test execution}, had $\bobprime$ not measured a query $(t_B, t_B')$. This claim and \Cref{thm:bbbv} imply that $(t_B, t_B') = (s,s')$ with non-negligible probability. \\

\par To prove this claim, suppose the opposite. We will describe a sequence of games, starting with \textbf{Game 1}, between $(\aliceprime, \widetilde{\bobprime})$ acting as the challenger and $(\alice, \bob)$ acting as the adversary: \begin{itemize}
    \item $\alice$ gets oracle access to $G', H'$ and gets input $\ket{A_{s,s'}}, v$, all of which are as sampled above by $\aliceprime$ and the MOE challenger.
    \item $\alice$ sends a quantum state $\rho_\B$ to $\bob$
    \item $\ati_{(P,Q), \gamma_1}^{\epsilon, \delta}$ as defined above, is applied to $(\bob, \rho_\B)$ by $\widetilde{\bob'}$, which uses the oracle $\oracletwo'$ without reprogramming alongside $(\oracleone_y')^A$, obtaining $b_B$. If $b_B = 0$, the game is aborted.
    \item A \textit{test execution} by $\widetilde{\bobprime}$ is run on $\bob$ with its leftover state and $\bob$ outputs $z$.
    \item The adversary wins if the output is correct, i.e. $z = f_y(x)$.
\end{itemize}

Note that since $\oracletwo'$ is not reprogrammed, the value $v$ is a random string independent from the rest of the game. Now we modify the game by replacing $\ket{A_{s,s'}}$ in the first step with the maximally mixed state, resulting in \textbf{Game 2}. The success probability of the adversary is unaffected due to the fact that the random strings $(s,s')$ only occur in $\ket{A_{s,s'}}$ in \textbf{Game 1}, and \[ \sum_{\substack{s \in \canonicalset(A) \\ s' \in \canonicalset(A^\perp)}}\ketbra{A_{s,s'}}{A_{s,s'}} = I
\]
for any subspace $A$.
\par Next, we replace the first oracle $(\oracleone_y')^A$ with a random oracle, obtaining \textbf{Game 3}. The success probability of the adversary again is affected only negligibly since $A$ is a random subspace independent of the rest of \textbf{Game 2}, which is statistically close to a random value by \Cref{clm:linearindependence}. Now, $y$ is an independent value from all of \textbf{Game 3} except for the test input $x_B$, hence the adversary is restricted to making a trivial guess, so that it cannot succeed with probability greater than $\trivialprob_B$. \\
\par Similarly, we argue that conditioned on $(t_B, t_B') = (s,s')$, the probability that $(t_C, t_C')$ is non-negligible. This follows by a similar argument after observing that after $\bobprime$ measures a query, the post-measurement state is still negligibly close to a state of the form \begin{align*}
    \sum_{j: \; \mu_j > \trivialprob_C + 1/4p} \theta_j \ket{\sigma_j}_\B \ket{\psi_j}_\C,
\end{align*}
for some states $\ket{\sigma_j}$, so that $\charlie$ will output correctly with probability greater than $\trivialprob_C + 1/4p$ during the final execution made by $\charlieprime$.

\end{proof}

\end{proof}

\par \Cref{lem:cphyb12,lem:cphyb23,lem:cphyb34,lem:cphyb4} together with triangle inequality imply that $p_1 \le \trivialprob + \negl(\secparam)$ as desired, finishing the proof of \Cref{lem:cpsecurity}.

\end{proof}}

\begin{remark} \label{rem:qptvsquery}
In our security proof, the adversary can run in unbounded time as long as it is query-bounded.
\end{remark}

\subversion{
\begin{remark}
Replicating the analysis in the proof of \Cref{thm:moe_rom}, our proof can be extended to the case when $\distrclass(y)$ samples correlated test inputs, i.e. the case when either $x_B = x_C = y$ or $x_B,x_C$ are both random.
\end{remark}
}

\fullversion{
\par Following techniques from the proof of \Cref{thm:moe_rom}, we can show security for correlated input distributions as well. 

\begin{corollary}
\label{cor:cp_correlated}
Let $\weight \in [0,1]$ and let $\distrclass^\weight(y)$ be the following input distribution: \begin{itemize}
    \item Sample $x_B, x_C \from \bit^\secparam \setminus \{y\}$ independently and uniformly at random.
    \item With probability $\weight$, output $(x_B, x_C)$.
    \item With probability $1 - w$, output $(y, y)$.
\end{itemize}
Then, $(\cp, \eval)$ above is a $\distrclass^w$-secure copy-protection scheme for point functions with input length $\secparam$.

\end{corollary}

\begin{proof}
Fix an unlearnable distribution $\distr_\secparam$ and define the following hybrids:

\begin{hybrid}{1} \label{hyb:cp_cor1}
This is the real piracy experiment for $(\cp, \eval)$.
\end{hybrid}

\begin{hybrid}{2}
\label{hyb:cp_cor2}
This hybrid matches \Cref{hyb:cp4} in the proof of \Cref{lem:cpsecurity}. In other words, we make the following changes: \begin{itemize}
    \item The oracles $\oracleone,\oracletwo$ for $\alice$ are replaced with reprogrammed oracles $\oracleone_y, \oracletwo_{s,s'}$, where $\oracleone_y(y)$ and $\oracletwo_{s,s'}(s,s')$ are reprogrammed to freshly random values.
    \item In addition, the oracle $\oracleone$ for $\bob$ and $\charlie$ both is changed to $\repryA$, where $\repryA(y)$ is reprogrammed to output a random (fixed) basis $(v_1, \dots, v_d)$ of $A$.
\end{itemize}
\end{hybrid}

Let $p_i$ be the probability that $\abc$ wins in Hybrid $i$. Note that $\trivialprob := \trivialprob(\distr_\secparam, \distrclass^w) = \max(w, 1 - \weight)$. By \Cref{lem:cphyb12,lem:cphyb23,lem:cphyb34}, we have $|p_1 - p_2| \le \negl(\secparam)$, since these lemmas are proved irrespective of the input distribution. Thus, it suffices to show that $p_2 \le \max(\weight, 1 - \weight)$.

\par Let $\rho_{\B\C}$ be the bipartite state created by $\alice$ in \Cref{hyb:cp_cor2}. Without loss of generality assume that $\rho_{\B\C} = \ketbra{\phi_{\B\C}}{\phi_{\B\C}}$ is a pure state. Fix $y \from \distr_\secparam$, fix $\oracleone_y, \oracletwo_{s,s'}, \oracletwo(s,s'), A$ which are randomly sampled, and fix random inputs $(x_B, x_C) \from \bit^\secparam \setminus \{y\} $. We define the following projectors: \begin{itemize}
    \item $\Pi^B_0$: runs $\bob$ on input oracles $\repryA, \oracletwo$ and test input $x_B$; it measures if the output is $f_y(x_B)$; then it undoes all the computation.
    \item $\Pi^B_1$: runs $\bob$ on input oracles $\repryA, \oracletwo$ and test input $y$; it measures if the output is $f_y(x_B)$; then it undoes all the computation.
    \item $\Pi^C_0$ and $\Pi^C_1$ are defined similarly for $\charlie$.
\end{itemize}

Now we write the state $\ket{\phi_{\B\C}}$ in its spectral decomposition with respect to $(\weight\Pi^B_0 + (1-\weight)\Pi^B_1) \otimes (\weight\Pi^C_0 + (1-\weight)\Pi^C_1)$ as \begin{align*}
    \ket{\phi_{\B\C}} = \sum_{i,j}\alpha_{i,j}\ket{\phi_i}_\B \ket{\psi_j}_\C,
\end{align*}
where $\ket{\phi_i}_\B$ is an eigenvector of $(\weight\Pi^B_0 + (1-\weight)\Pi^B_1)$ with eigenvalue $\lambda_i$ and $\ket{\psi_j}_\C$ is an eigenvector of $(\weight\Pi^C_0 + (1-\weight)\Pi^C_1)$ with eigenvalue $\mu_j$.

\end{proof}

\par We first make the following observation:

\begin{lemma} \label{lem:cp_cor_smallweight}
    Let $p(\cdot)$ be a polynomial. With overwhelming probability over $y,(s,s'),\allowbreak \oracleone_y,\allowbreak \allowbreak \oracletwo_{s,s'},\allowbreak \oracletwo(s,s'),A$, and $(x_B,x_C)$, we have \begin{align*}
        \sum_{\substack{i: \; |\lambda_i - 1/2| > |\weight - 1/2| + 1/p \\ j: \; |\mu_j - 1/2| > |\weight - 1/2| + 1/p}} \abs{\alpha_{i,j}}^2 \le \negl(\secparam).
    \end{align*}
\end{lemma}

\begin{proof}
Note that the condition $|\lambda_i - 1/2| > |\weight - 1/2| + 1/p$ is satisfied if and only if $\lambda_i > \trivialprob + 1/p$ or $1 - \lambda_i > \trivialprob + 1/p$.  The proof is nearly identical to the proof of \Cref{lem:cpsmallweight}. To avoid repetition, we only mention a few notable differences: \begin{itemize}
    \item After sampling $y \from \distr_y$, $\alice'$ additionally samples random inputs $x_B \ne y$ and $x_C \ne y$.
    \item Instead of $\ati$, $\bob'$ applies $\sati^{\epsilon,\delta}_{P,Q,\gamma_1}$, with $P = \weight\Pi^B_0 + (1 - \weight)\Pi^B_1$, $Q = I-P$, $\gamma_1 = 3/4p$, and $\epsilon = 1/2p$. Similarly for $\charlie'$.
    \item When implementing $\Pi^B_0$, $\bob'$ uses $x_B$ as input, and it uses $y$ when implementing $\Pi^B_1$. Similarly for $\charlie'$.
    \item In the end when we say that an adversary, with no knowledge of $y$ other than the test input given, can succeed with probability at most $\trivialprob$, we instead argue the success probability of such an adversary, denoted by $q$, must satisfy $\max(q, 1-q) \le \trivialprob$. This is because the adversary can always flip its output bit to succeed with probability $q$ instead of $1 - q$.
\end{itemize}
\end{proof}

\par By \Cref{lem:cp_cor_smallweight}, with overwhelming probability $\ket{\phi_{\B\C}}$ is negligibly close to the state $\ket{\phi_\bob'} + \ket{\phi_{\charlie}'}$, where \begin{align*}
    \ket{\phi_\bob'} &= \sum_{i\; :\; |\lambda_i - 1/2| \le |\weight - 1/2| + 1/p} \alpha_{i,j}\ket{\phi_i}_\B \ket{\psi_j}_{\C}, \\
    \ket{\phi_\charlie'} &= \sum_{\substack{i\; :\; |\lambda_i - 1/2| > |\weight - 1/2| + 1/p \\ j\; :\; |\mu_j - 1/2| \le |\weight - 1/2| + 1/p}} \alpha_{i,j}\ket{\phi_i}_\B \ket{\psi_j}_{\C}.
\end{align*} The rest of the proof will imitate the analysis in the proof of \Cref{lem:rprom}:
    \begin{align*}
        &\weight\left| (\Pi^B_0 \otimes \Pi^C_0) (\ket{\phi'_\Bs} + \ket {\phi'_\Cs})\right|^2 + (1-\weight)\left| (\Pi^B_1 \otimes \Pi^C_1) (\ket{\phi'_\Bs} + \ket {\phi'_\Cs)}\right|^2 \\
        =& \big( \weight\langle \phi'_{\Bs} | (\Pi^B_0 \otimes \Pi^C_0) | \phi'_{\Bs}  \rangle +  (1-\weight)\langle \phi'_{\Bs} | (\Pi^B_1 \otimes \Pi^C_1) | \phi'_{\Bs}  \rangle  + \weight\langle \phi'_{\Cs} | (\Pi^B_0 \otimes \Pi^C_0) | \phi'_{\Cs}  \rangle \\ +& (1-\weight) \langle \phi'_{\Cs} | (\Pi^B_1 \otimes \Pi^C_1) | \phi'_{\Cs}  \rangle \big)
        + 2\mathsf{Re}\left(  \weight\langle \phi'_{\Bs} | (\Pi^B_0 \otimes \Pi^C_0) | \phi'_{\Cs} \rangle + (1-\weight)\langle \phi'_{\Bs} | (\Pi^B_1 \otimes \Pi^C_1) | \phi'_{\Cs}  \rangle \right) \\
        \leq & \big( \weight\langle \phi'_{\Bs} | (\Pi^B_0 \otimes I) | \phi'_{\Bs}  \rangle + (1-\weight) \langle \phi'_{\Bs} | (\Pi^B_1 \otimes I) | \phi'_{\Bs}  \rangle  + \weight\langle \phi'_{\Cs} | (I \otimes \Pi^C_0) | \phi'_{\Cs}  \rangle \\ +& (1-\weight) \langle \phi'_{\Cs} | (I \otimes \Pi^C_1) | \phi'_{\Cs}  \rangle \big)
        + 2\mathsf{Re}\left(  \weight\langle \phi'_{\Bs} | (\Pi^B_0 \otimes \Pi^C_0) | \phi'_{\Cs} \rangle + (1-\weight)\langle \phi'_{\Bs} | (\Pi^B_1 \otimes \Pi^C_1) | \phi'_{\Cs}  \rangle \right).
    \end{align*}
    We bound each term separately. 
    \begin{itemize}
        \item $\left(\weight \langle \phi'_{\Bs} | (\Pi^B_0 \otimes I) | \phi'_{\Bs}  \rangle + (1-\weight) \langle \phi'_{\Bs} | (\Pi^B_1 \otimes I) | \phi'_{\Bs}  \rangle \right)$. It is equal to $\langle \phi'_{\Bs} | (\weight\Pi^B_0 + (1-\weight)\Pi^B_1) \otimes I | \phi'_{\Bs} \rangle$; by the definition of $\ket{\phi'_\Bs}$, it will be at most $$(1/2 + |\weight - 1/2| + 1/p) |\ket{\phi'_\Bs}|^2 = \max(\weight,1-\weight)|\ket{\phi'_\Bs}|^2.$$ 
        
        \item $\left(\weight \langle \phi'_{\Cs} | (\Pi^C_0 \otimes I) | \phi'_{\Cs}  \rangle + (1-\weight) \langle \phi'_{\Cs} | (\Pi^C_1 \otimes I) | \phi'_{\Cs}  \rangle \right)$. Similar to the above case, it is at most \allowbreak $\max(\weight,1-\weight)|\ket{\phi'_\Bs}|^2$.
        
        \item $\mathsf{Re}\left(  \langle \phi'_{\Bs} | (\Pi^B_0 \otimes \Pi^C_0) | \phi'_{\Cs} \rangle \right)$.  By \Cref{cor:orthogonal_eigen}, this term will vanish: 
        \begin{align*}
            \langle \phi'_{\Bs} | (\Pi^B_0 \otimes \Pi^C_0) | \phi'_{\Cs} \rangle &= \sum_{i\; :\; |\lambda_i - 1/2| \le |\weight - 1/2| + 1/p}  \sum_{\substack{i'\; :\; |\lambda_i - 1/2| > |\weight - 1/2| + 1/p \\ j'\; :\; |\mu_j - 1/2| \le |\weight - 1/2| + 1/p}}  \alpha^\dagger_{i, j} \alpha_{i', j'} \langle \phi_i | \Pi^B_0 | \phi_{i'} \rangle \langle \psi_j | \Pi^C_0 | \psi_{j'} \rangle; 
        \end{align*}
        since every possible $i, i'$ satisfy $\lambda_i + \lambda_{i'} \ne 1$, we have $\langle \phi_i | \Pi^B_0 | \phi_{i'} \rangle = 0$.
        
        \item $\mathsf{Re}\left(  \langle \phi'_{\Bs} | (\Pi^B_1 \otimes \Pi^C_1) | \phi'_{\Cs} \rangle \right)$.  Similarly, this term vanishes as well. 
    \end{itemize}
    
    Therefore, the total probability is at most $$\weight\left| (\Pi^B_0 \otimes \Pi^C_0) (\ket{\phi'_\Bs} + \ket {\phi'_\Cs})\right|^2 + (1-\weight)\left| (\Pi^B_1 \otimes \Pi^C_1) (\ket{\phi'_\Bs} + \ket {\phi'_\Cs)}\right|^2 + \negl(n) \leq \max(\weight,1-\weight) + \frac{1}{p} + \negl(n).$$
    
    \medskip
    
    Since the polynomial $p(\cdot)$ is arbitrary, this suffices for the proof.
} 

\printbibliography

\subversion{
\appendix
\newpage
\section*{\huge{Supplementary Material}} %
\section{Testing Quantum Programs}

\paragraph{Efficient Measurement.}
The projective measurement $\cE$ above is not efficiently computable in general. 
However, they can be approximated if the POVM is a mixture of projective measurements, as shown by Zhandry \cite{z20}, using a technique first introduced by Marriott and Watrous \cite{marriott2005quantum}. 

Consider the following procedure as a binary POVM $\cP = (P, Q)$ acting on a quantum program $\rho$: samples a random challenge $r$, evaluates the program on $r$, and checks if the output is correct. 
This procedure can be viewed as (1). picking a uniformly random challenge $r$; (2). applying a \emph{projective measurement} $U_r$. In this case, $P = \frac{1}{R} \sum_r U_r$ where $R$ is the size of the challenge space. 
This POVM captures the situation where a challenger randomly samples a classical challenge and tests if a quantum program's classical outcome is correct on that challenge.

\section{Missing Proofs} \label{sec:proofs}
\subsection{Correctness of Copy-Protection} \label{sec:cp_correctness}

\subsection{Security of Copy-Protection} \label{sec:cp_security}

}

\end{document}